\newif
\iftwocol
\twocoltrue

\iftwocol
\documentclass[10pt,twocolumn,twoside]{IEEEtran}
\else
\documentclass[12pt,draft,final,onecolumn]{article}
\fi

\usepackage[sumlimits]{amsmath}
\usepackage{amstext}
\usepackage{amsfonts}
\iftwocol
\else
\usepackage{amsthm}
\fi
\usepackage{amssymb}
\usepackage{latexsym}
\usepackage{graphicx}
\usepackage{setspace,graphicx}
\usepackage{color}
\usepackage{subfigure}
\usepackage{cite}

\newtheorem{theorem}{Theorem}
\newtheorem{definition}{Definition}

\newtheorem{notation}{Notation}
\newcommand{\D}{\mathcal{D}}
\newcommand{\BIT}{\begin{itemize}}
\newcommand{\EIT}{\end{itemize}}
\newcommand{\BEN}{\begin{enumerate}}
\newcommand{\EEN}{\end{enumerate}}
\newcommand{\mc}[1]{\mathcal{#1}}
\newcommand{\mb}[1]{\mathbf{#1}}

\newcommand{\Ev}{\mathcal{E}}
\newcommand{\I}{\mathcal{I}}
\newcommand{\J}{\mathcal{J}}
\newcommand{\T}{\mathcal{T}}

\newcommand{\M}{\mathcal{M}}
\newcommand{\B}{\mathcal{B}}
\newcommand{\A}{\mathcal{A}}
\newcommand{\C}{\mathcal{C}}

\newcommand{\yb}{\mathbf{y}}
\newcommand{\xb}{\mathbf{x}}
\newcommand{\ra}{\rightarrow}
\newcommand{\EX}{{\noindent\em {Example: }}}
\newcommand{\EEX}{\hspace*{\fill}~\QED\par\endtrivlist\unskip}
\newcommand{\nra}{\nrightarrow}

\begin{document}

\title{Parity Forwarding for Multiple-Relay Networks
\thanks{
Manuscript has been submitted to the {\it IEEE Transactions on
Information Theory} on November 12, 2007. The materials in this
paper have been presented in part at the IEEE International
Symposium on Information Theory (ISIT), Seattle, WA, U.S.A., July
2006, and in part at the IEEE International Symposium on
Information Theory (ISIT), Nice, France, June 2007. The authors are
with The Edward S. Rogers Sr. Department of Electrical and
Computer Engineering, University of Toronto, 10 King's College
Road, Toronto, Ontario M5S 3G4, Canada. e-mails:
peyman@comm.utoronto.ca, weiyu@comm.utoronto.ca. Phone:
416-946-8665. FAX: 416-978-4425. Kindly address correspondence
to Peyman Razaghi (peyman@comm.utoronto.ca). }}
\author{Peyman Razaghi, {\it Student Member, IEEE}, and
    Wei Yu, {\it Member, IEEE}}
\date{\today}
\maketitle

\begin{abstract}
This paper proposes a relaying strategy for the multiple-relay
network  in which each relay decodes a selection of transmitted
messages by other transmitting terminals, and forwards parities of
the decoded codewords. This protocol improves the previously
known achievable rate of the decode-and-forward (DF) strategy for
multirelay networks by allowing relays to decode only a selection of
messages from relays with strong links to it. Hence, each relay may
have several choices as to which messages to decode, and for a
given network many different parity forwarding protocols may exist.
A tree structure is devised to characterize a class of parity
forwarding protocols for an arbitrary multirelay network. Based on
this tree structure, closed-form expressions for the achievable rates
of these DF schemes are derived. It is shown that parity forwarding
is capacity achieving for  new forms of degraded relay networks.
\end{abstract}

\normalsize

\section{Introduction}
A relay network consists of a pair of source and destination
terminals and a number of { relays}. The relays have no message of
their own and only help the source communicate to the destination.
Fig.~\ref{fig:multirelays} shows a schematic of a network with $K$
relays in which the relays are numbered from 1 to $K$, the source is
represented by index 0, and the destination  is represented by index
$K+1$. The random variables $X_0,X_1,\ldots,X_K$ represent the
transmitted signals, and $Y_1,Y_2,\ldots,Y_{K+1}$ represent the
received signals, at respective nodes. The channel is assumed to be
memoryless and is defined by the joint probability distribution
function (pdf) $p(y_1,y_2,\ldots,y_{K+1}|x_0,x_1,\ldots,x_{K})$.

Although the capacity of the simple yet fundamental
single-relay network introduced by Van der Meulen in
\cite{van_der_meulen} is still open, the  recent surge of
interests in relay networks has resulted in new communication
protocols and achievable rates for multirelay networks
\cite{xie_kumar_rate,gupta_kumar_rate,kramer_gastpar_gupta,reznik_kulkarni_verdu,gastpar_vetterli,rost_generalized_2007}.
Among the various relaying strategies, the classical
decode-and-forward (DF) strategy, proposed by Cover and El
Gamal \cite{cover_elgamal}, has been of  particular interest.
In the DF scheme for the single-relay channel,   the relay
decodes the source message and forwards  a bin index for it to
the destination. This fundamental relaying strategy is proved
to be capacity achieving for a degraded single-relay network
\cite[Theorem~1]{cover_elgamal}.

Generalizations of the single-relay DF scheme to the multirelay case
have  been studied in
\cite{aref_thesis,gupta_kumar_rate,kramer_gastpar_gupta,xie_kumar_rate,gastpar_vetterli,reznik_kulkarni_verdu}.
The best known DF strategy  for multirelay networks is called the
multihop protocol, devised in \cite{xie_kumar_rate}. In the
multihop scheme, the source  and a group of relays that have
already decoded the source message replicate and cooperatively
transmit the source message to the next relay. This process is
repeated until all relays decode the source message and
cooperatively transmit the source message to the destination. The
decoding procedures at the receivers take into account that the
message is transmitted over several blocks. It has been proved that
multihopping along with optimal decoding is capacity achieving for
the generalized multirelay version of the single-relay degraded
channel \cite{xie_kumar_rate}.

This paper shows that the multihop protocol can be further
improved. The main bottleneck of the multihop relaying strategy is
that all relay terminals must decode the source message in order to
participate in the relaying protocol. This can be restrictive, because
the source rate is constrained by the decodability conditions at
those relays with poor links from the source. However, it is not
necessary to require all relays to decode the source message. The
multihop DF rate is improved if relays are allowed to choose an
appropriate set of messages to decode. This set of messages may
include not only the source message, but also messages from other
relays. This flexibility can significantly improve the DF rate.

\begin{figure}
  \centering
  \includegraphics[scale=0.5]{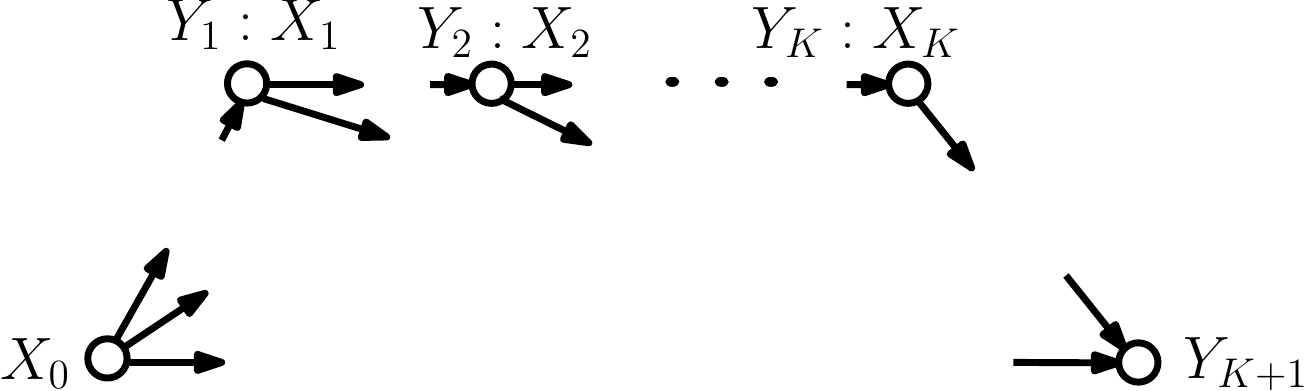}
  \caption{A general network with multiple relays.\label{fig:multirelays}}
\end{figure}

In this paper, a class of DF protocols, named parity forwarding, is
proposed to improve the DF rate for the multirelay network. In
parity forwarding, the source and the relays may transmit multiple
messages. Relays decode a selection of messages transmitted by
other nodes, called the decoding set. The messages sent by  relays
are bin indices containing partial information for the messages in
the decoding set.  A key feature of the parity forwarding protocol as
compared to the multihop protocol is that the relay messages
partially describe the decoded data at the relays, rather than fully
replicate the decoded data. To decode a message, the decoders
(i.e., the final destination or intermediate relays) identify all
messages that contain partial information about the message to be
decoded, and perform joint decoding by combining all partial
information.

For a given multirelay network, several parity forwarding protocols
are possible, depending on the messages decoded or transmitted by
the relays. This paper characterizes a class of parity forwarding
protocols in an structured way via a message tree. The message
tree characterizes the dependencies between messages in the
network. Using this tree structure, the appropriate joint decoding
procedures at receiver nodes are identified and closed-form
expressions  for the achievable rates are derived. Further, it is
shown that under certain degradedness conditions, the rates
achievable by parity forwarding are the capacities.



The proposed relaying scheme is named parity forwarding, because
in a linear coding context, forwarding message bin indices by the
relays is equivalent to forwarding parity bits. This is because parity
bits of a linear code  partition the linear codebook into subcodes,
which are analogous to bins \cite{wyner_binning}. We use the
terms ``bin index'' and ``parity message'' interchangeably
throughout the paper. In a related work \cite{peyman_wei_it}, the
interpretation of bins as parities  also allowed the design of
practical capacity-approaching codes for the single-relay and certain
forms of multirelay networks.

The parity forwarding strategy proposed in this paper bears a
resemblance to network coding \cite{li_yeung_cai}. In both
schemes, the intermediate nodes forward parities, instead of
replicating decoded messages. To decode messages, the embedded
information in parities is combined at each receiver.

The rest of the paper is organized as follows:  We begin with a
review of existing single-relay DF schemes in Section
\ref{sec:1relay}, and introduce joint decoding for the
single-relay network. The parity forwarding protocol is
introduced in Section \ref{sec:2relays} for a network with two
relays. The general multiple-relay parity forwarding protocol
is described in Section \ref{sec:multirelay}. Section
\ref{sec:examples} illustrates key features of the parity
forwarding protocol through several examples. The capacity of
new forms of multirelay networks are also derived in Section
\ref{sec:examples}. Finally, Section \ref{sec:conclusions}
concludes the paper.

\section{Single-relay DF: A joint decoding approach \label{sec:1relay}}
The DF scheme was first introduced in \cite{cover_elgamal} for the
single-relay channel and was shown to be capacity achieving for the
degraded  relay channel.  Since then, several variants of DF have
been developed and extended to multirelay networks
\cite[Ch.~4]{aref_thesis},
\cite{gupta_kumar_rate,gupta_capacity_2000,schein_gallager,gastpar_kramer_gupta,xie_kumar_scaling,xie_kumar_rate,reznik_kulkarni_verdu}.
In \cite{xie_kumar_rate},  the single-relay DF strategy is extended
to multirelay networks and is shown to attain the capacity of the
generalized multirelay version of the degraded single-relay channel.
In this section, a single-relay DF approach based on joint decoding
is proposed. This new approach combines the advantages of the DF
methods in \cite{cover_elgamal} and \cite{xie_kumar_rate} and
allows us to further improve the multirelay DF rate.

\subsection{ Regular Encoding vs. Irregular Encoding}
There are several variations of the single-relay DF scheme,
depending on their respective encoding and decoding methods. The two
main encoding methods are the regular encoding approach of
\cite{xie_kumar_rate} (first proposed in \cite{carleial} for a
different channel) and the DF encoding approach originally
introduced in \cite{cover_elgamal}, which is later named irregular
encoding in \cite{kramer_gastpar_gupta}.

A key discriminating feature of regular encoding and irregular
encoding is the rate of the relay message. In regular encoding,
the relay message rate is equal to the source message rate,
whereas in irregular encoding, the relay message rate can be
smaller than the source message rate.

In both methods, encoding is performed blockwise. Let
$m^t_0\in\{1,2,\ldots,2^{nR_0}\}$ and
$m^t_1\in\{1,2,\ldots,2^{nR_1}\}$ be the source and the relay
messages in block $t$. In regular encoding, the relay message
in block $t$ is equal to the source message in block $t-1$,
i.e., $m^t_1=m^{t-1}_0$, hence the relay message rate is
limited to $R_1=R_0$.  On the other hand, in irregular
encoding, $m^t_1$ is a random bin index for $m^{t-1}_0$, which
allows for more encoding flexibility. The bin index $m^t_1$ is
computed according to $m^t_1=P_{\B_1}(m^{t-1}_0)$, where
$P_{\B_1}(\, \cdot \,)$ is the binning function and $\B_1$ is a
uniform random partition of $\{1,2,\ldots,2^{nR_0}\}$ as
defined in the following.

\begin{definition}[Binning Function] Let
$\mc{B}_y=\{\mc{S}_1,\mc{S}_2,\cdots,\mc{S}_{2^{nR_y}}\}$ be a
uniform random partition of $\{1,2,\cdots,2^{nR_x}\}$ into
$2^{nR_y}$ bins $\mc{S}_k$ of size $2^{n(R_x-R_y)}$ indexed by
$\mc{Y}=\{1,2,\cdots,2^{nR_y}\}$. The  binning function
$P_{\mc{B}_y}(\, \cdot \,)$ returns the bin index of its argument
with respect to  $\mc{B}_y$, i.e., $v=P_{\mc{B}_y}(u)$ if and only
if $u\in \mc{S}_v$.
\end{definition}

Codebook construction is the same for both regular and irregular
encoding schemes. At the relay, $2^{nR_1}$ random codewords
$\mb{x}_1(m_1)$  of length $n$ are generated according to
$p(x_1)$ to encode $m^t_1$ in block $t$. The source codebook is
constructed using superposition encoding  to encode both $m^t_0$
and $m^t_1$ in block $t$\cite{cover_elgamal}. This is to allow the
source to cooperate with the relay, as the source in each block
knows the message of the relay. More specifically, the source
generates $2^{nR_0}$ codebooks, one for every $\mb{x}_1(m_1)$
codeword. For each codeword $\mb{x}_1(m_1)$, $2^{nR_0}$
codewords $\mb{x}_1(m_0|m_1)$ are randomly generated according
to $p(x_0|x_1)$.

Both schemes give the same DF rate for the single-relay channel.
However, the generalization of DF to multirelay networks is more
straightforward for regular encoding, since the relay messages are
replications of the source message. This results in the multihop
scheme. On the other hand, irregular encoding potentially allows for
DF schemes other than multihop relaying because of its flexibility
with respect to the relay messages. To the best of our knowledge,
multirelay DF scheme based on irregular encoding has not been
proposed prior to this work, due to limitations of successive
decoding when used along with irregular encoding in multirelay
networks.

\subsection{Successive Decoding vs. Window Decoding}
Successive decoding \cite{cover_elgamal} and window decoding
\cite{carleial} are the corresponding decoding approaches for
the irregular and regular encoding methods, respectively (see
\cite{kramer_gastpar_gupta} for a detailed summary of DF
decoding approaches). In successive decoding, the destination
first decodes the relay message, then decodes the source
message with the help of the decoded relay message. The
resulting constraints on the rates of the source and the relay
messages to ensure successful decoding  at the destination are
summarized below \cite{cover_elgamal}:
\begin{subequations}\label{eq:1relay:tx-rx:bound}
\begin{align}
R_1&\leq I(X_1;Y_2) \label{eq:1relay:tx-rx:bound:1}\\
R_0&\leq I(X_0,X_1;Y_2)\\
R_0&\leq I(X_0;Y_2|X_1)+R_1
\end{align}
\end{subequations}

When generalizing to multiple-relay networks, successive decoding
of messages is restrictive if the relay message is to be decoded at
multiple receivers (e.g., at another relay and at the final
destination). In this case, because the downlink channels from the
relay to different downstream receivers have different capacities,
the rate of the relay message must be smaller than the minimum of
the downlink capacities to ensure successful decoding of the relay
message at all intended receivers (e.g., see
\cite[(2)-(4)]{gupta_kumar_rate} for constraints of this type).
However, this is not optimal, since the extra rates of the downlink
channels with higher capacities are wasted.

The rate limitation problem of successive decoding in multirelay
networks is resolved in  the multihop strategy by using regular
encoding and window decoding (see \cite{xie_kumar_rate} for the
details). In regular encoding, all relays transmit replications of the
source message, thus the rates of all messages transmitted over
downlink channels are the same as the source rate (thus, no
channel operates at a rate below its capacity). The source message
is decoded by observing the received sequences over a window of
successive blocks. However, the multihop protocol is not the only
possible DF protocol for multiple-relay networks. The next
subsection describes a joint decoding approach that allows for more
flexible multirelay DF methods.

\subsection{Joint Decoding for Irregular Encoding \label{sec:1relay:joint_dec}}
Irregular encoding corresponds to forwarding bin indices for
the received messages at the relay terminal. The key element
that allows irregular encoding to be generalized to multirelay
networks is a joint decoding procedure that avoids the
shortcomings of successive decoding. To illustrate joint
decoding, we consider the single-relay channel in this section.
Joint decoding for  a multirelay network follows the same
principle. Note that in contrast to the multirelay networks,
joint decoding has no effect on the single-relay DF rate.

The decoding procedure at the relay is similar to the one in
\cite{cover_elgamal} or \cite{xie_kumar_rate} (details are omitted
for brevity). The relay in each block decodes the source message
provided that the source rate satisfies the following constraint.
\begin{equation}\label{eq:1relay:tx_re:bound} R_0\leq
I(X_0;Y_1|X_1).
\end{equation}

The destination jointly decodes the pair of messages $m^{t-1}_0$
and $m^t_1$ over the two successive blocks $t-1$ and $t$. Assume
that in block $t$, the destination has already decoded $m^{t-2}_0$
and $m^{t-1}_1$ correctly. (It will become clear later that this is a
valid assumption.) Knowing $m^{t-1}_1$ in block  $t$, the
destination finds a pair of messages $m_0$ and $m_1$ satisfying
$m_1=P_{\B_1}(m_0)$, such that given $\xb_1(m_1)$,
$\xb_0(m_0|m^{t-1}_1)$ is jointly typical with $\mb{y}_2^{t-1}$,
the received sequence in block $t-1$,  and $\xb_1(m_1)$ is jointly
typical with $\mb{y}_2^t$, the received sequence in block $t$.

The probability that an incorrect $\xb_0$ is jointly typical with
$\mb{y}^{t-1}_2$ given $\xb_1$ is asymptotically equal to
$2^{-nI(X_0;Y_2|X_1)}$ \cite[Theorem~15.2.3]{cover_elements};
similarly, the probability that an incorrect $\xb_1$ is jointly typical
with $\mb{y}^t_2$ is asymptotically bounded by
$2^{-nI(X_1;Y_2)}$. Let $\Ev_1$ denote the event that $\xb_0$ is
decoded incorrectly, and $\Ev_2$ be the event that $\xb_1$ is
decoded incorrectly. The decoding error probability is given by
$\text{Pr}(\Ev_1)$, which can be bounded as
$\text{Pr}(\Ev_1)=\text{Pr}\left(\Ev_1\cap(\Ev_2\cup\Ev_2^c)\right)\leq
\text{Pr}(\Ev_1\cap\Ev_2)+\text{Pr}(\Ev_1\cap\Ev_2^c)$. Now,
$\text{Pr}(\Ev_1\cap\Ev_2)$ is asymptotically bounded by
$2^{nR_0}2^{-nI(X_0;Y_2|X_1)}2^{-nI(X_1;Y_2)}$. This is
because $m_1$ is a function of $m_0$, hence there are $2^{nR_0}$
pairs of $m_0$ and $m_1$ messages in total. On the other hand,
$\text{Pr}(\Ev_1\cap\Ev_2^c)$ is asymptotically bounded by
$2^{n(R_0-R_1)}2^{-nI(X_0;Y_2|X_1)}$, since knowing $m_1$
(i.e., $\Ev_2^c $ has occurred), there remain $2^{n(R_0-R_1)}$
choices for $m_0$. Hence, the decoding error probability at the
destination  asymptotically approaches zero if
\begin{subequations}\label{eq:1relay:tx-rx:bound:joint}
\begin{align}
\iftwocol
R_0&\leq I(X_0;Y_2|X_1)+I(X_1;Y_2)\\
&=I(X_0,X_1;Y_2)\nonumber\\
\else
\intertext{and}
R_0&\leq I(X_0;Y_2|X_1)+I(X_1;Y_2)=I(X_0,X_1;Y_2)\\
\fi
R_0&\leq I(X_0;Y_2|X_1)+R_1
\end{align}
\end{subequations}
Note that the rate of the relay message $R_1$ appears only on the
right-hand side of \eqref{eq:1relay:tx-rx:bound:joint} and thus is
not constrained. The fact that the rate of the relay message is
unconstrained is the key advantage of joint decoding as compared
to successive decoding.

Joint decoding along with irregular encoding combines the
benefits of irregular encoding and regular encoding by
providing rate flexibility for relay messages.
Fig.~\ref{fig:joint_irregular}\ describes the advantage of
combining
 joint decoding and irregular encoding. In successive decoding along
with irregular encoding, the relay message rate $R_1$ must satisfy
$R_1\leq I(X_1;Y_2)$. On the other hand, in window decoding along
with regular encoding,  the relay message has to be equal to the
source message, which only allows multihop type of schemes and
restricts $R_1$ to be equal to $R_0$. The combination of joint
decoding and irregular encoding allows the relay messages to have
any rate $R_1$ satisfying $I(X_1;Y_2)\leq R_1\leq R_0$. This
flexibility in choosing the rate of relay messages is the key to
extend multirelay DF beyond the multihop scheme.

\begin{figure}
\includegraphics[scale=0.56]{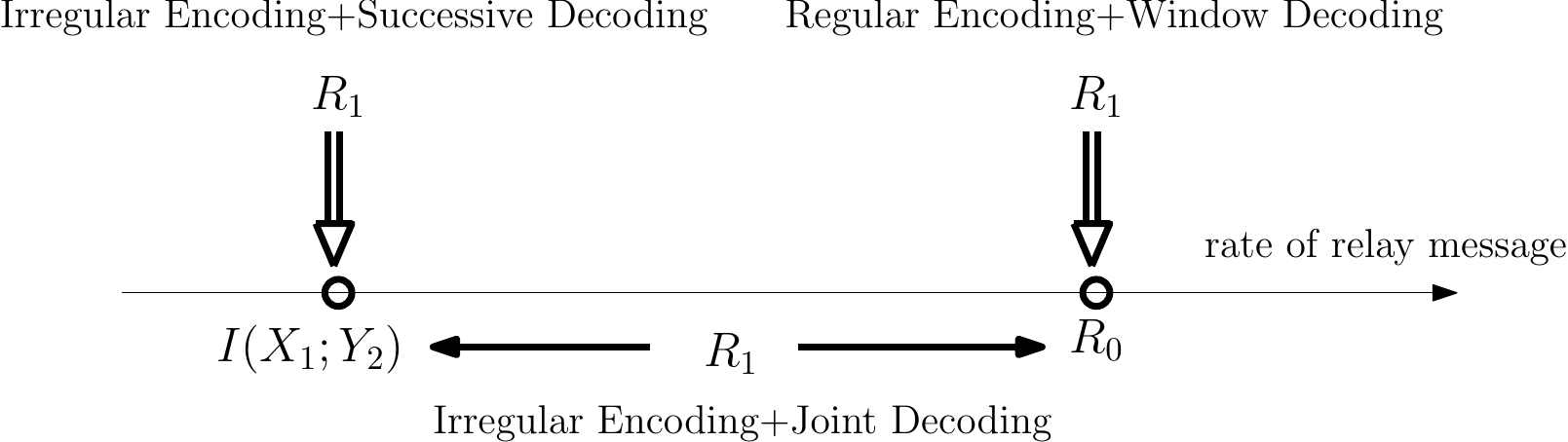}
\caption{If joint decoding of the source and the relay message is
performed at the destination, the rate of the relay message can be
flexible. This flexibility is the key to obtain higher DF rates in a
multiple-relay network. \label{fig:joint_irregular}}
\end{figure}

\section{Parity Forwarding in a Two-Relay Network \label{sec:2relays}}
We begin our discussion of the multirelay network by providing two
examples of parity forwarding DF protocols for a network with two
relays.  The first protocol, named Protocol A, demonstrates that
irregular encoding along with joint decoding achieves the best
previous multirelay DF rate, obtained via regular encoding in
\cite{xie_kumar_rate}. The second protocol, named Protocol B,
demonstrates that the DF rate of \cite{xie_kumar_rate} can be
further improved. Protocol B also identifies the capacity of a new
degraded form of two-relay networks. Later in Section
\ref{sec:examples}, a third two-relay DF protocol with a different
achievable rate is described as a more involved example of the
parity forwarding protocol. This third protocol is also capacity
achieving for a class of two-relay networks under certain
determinism and degradedness conditions as discussed in Section
\ref{sec:examples}. In general,  many different DF protocols are
possible  in a multirelay network. Subsequent sections describe a
structured characterization of a variety of DF protocols for an
arbitrary multirelay network.

\begin{figure}[t]
  \begin{center}
  \iftwocol
      \scalebox{0.5}{\includegraphics{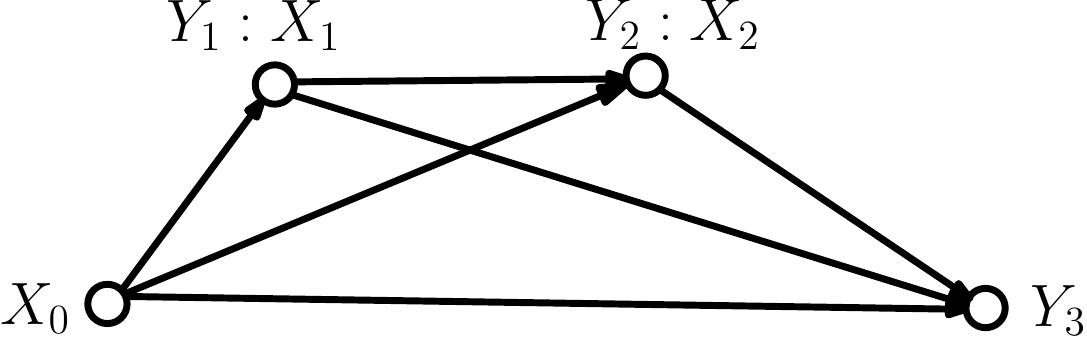}}
  \else
        \scalebox{0.7}{\includegraphics{2relays.pdf}}
  \fi
      \caption{A general two-relay network.\label{fig:2relays}}
  \end{center}
\end{figure}

\subsection{Encoding}
Fig.~\ref{fig:2relays} shows the schematic of a network with two
relays. Let $m_0$ denote the source message, $m_1$ the message
transmitted by the first relay, and $m_2$ the message of the
second relay.  In  Protocol A, the first relay decodes the source
message  $m_0$ and forwards $m_1$, a random bin index for
$m_0$, to help the second relay decode $m_0$; the second relay
decodes $m_0$ with the help of $m_1$ and forwards $m_2$, a
random bin index for $m_1$; the destination decodes $m_0$ with
the help of both $m_1$ and $m_2$. In Protocol B, again the first
relay decodes the source message $m_0$ and forwards $m_1$, a bin
index for $m_0$, to the second relay; the second relay decodes
{\em only} $m_1$, without attempting to decode $m_0$, and
forwards $m_2$, a random bin index for $m_1$, to the destination;
the destination decodes the source message $m_0$ with the help of
$m_1$ and $m_2$, both functions of $m_0$.

The difference between the two protocols is that in Protocol A, the
second relay decodes the source message $m_0$, whereas in
Protocol B, the second relay only  decodes $m_1$, the message of
the first relay. This difference results in different achievable DF
rates for the two protocols. This  can be best observed if  we
consider the following extreme cases: if the channel from the source
to the second relay is completely blocked and the source can only
communicate to the second relay through the first relay, Protocol A
achieves a lower rate as compared to Protocol B, since in Protocol A
the source data rate is bounded by the decodability condition of the
source message at the relay, whereas the rate of the source
message in Protocol B is not so constrained. However, if the
channel from the source to the second relay is strong, then Protocol
A may give a higher rate, since in Protocol A, the second relay can
take advantage of $m_1$ when decoding $m_0$, whereas in
Protocol B, the second relay receives no extra help when  decoding
$m_1$.


The encoding scheme is the same for both protocols. A binning
scheme along with block Markov encoding is performed. Let $n$
denote the length of each transmission block. In block $t$, the
messages of the source, the first relay, and the second relay,
$m^t_0$, $m^t_1$, and $m^t_2$ are selected from the sets
$\{1,2,\ldots,2^{nR_0}\}$, $\{1,2,\ldots,2^{nR_1}\}$, and
$\{1,2,\ldots,2^{nR_2}\}$, respectively. In each block, the
message of the first relay is a random bin index (parity message)
for the previous source message, and the message of the second
relay is a parity message for the previous message of the first
relay. The three messages in the network are related as
$m^t_1=P_{\B_1}(m^{t-1}_0)$, and
$m^t_2=P_{\B_2}(m^{t-1}_1)$, where $\B_1$ and $\B_2$ are
random partitions of sizes $2^{nR_1}$ and $2^{nR_2}$,
respectively.

Random codebooks are used to encode the three messages. To
encode $m_2$ at the second relay, $2^{nR_2}$ random codewords
$\mb{x}_2(m_2)$ of length $n$ are generated according to
$p(x_2)$. The codebook used at the first relay depends on the
codeword transmitted by the second relay. This is to allow the first
relay to cooperate with the second relay. As the first relay in each
block knows the message of the second relay, it can generate
$2^{nR_2}$ codebooks, one for each one of $\mb{x}_2(m_2)$
codewords. To generate a codebook depending on the codeword
$\mb{x}_2(m_2)$, $2^{nR_1}$ codewords $\mb{x}_1(m_1|m_2)$
are randomly generated according to $p(x_1|x_2)$. Similarly, in
order to allow the source to cooperate with the messages of the
first and the second relay, which are known to the source in each
block, the source codebook is constructed by generating
$2^{nR_0}$ codewords $\mb{x}_0(m_0|m_1,m_2)$ for each pair of
$m_1$ and $m_2$ messages, according to  $p(x_0|x_1,x_2)$. This
is an instance of {\em superposition broadcast} or briefly {\em
superposition} encoding, as it is similar to the superposition
codebook construction for the degraded broadcast channel
\cite[Chapter~14]{cover_elements}.


\subsection{Decoding}
Decoding at the relay and at the destination is performed using
joint typicality test. In block $t$, the first relay knows $m^t_1$ and
$m^t_2=P_{\B_2}(m^{t-1}_1)$. To decode $m^t_0$, the first relay
finds a codeword $\xb_0(m_0|m^t_1,m^t_2)$ that is jointly typical
with its received sequence  $\mb{y}^t_1$ given $\xb_2(m^t_2)$
and $\xb_1(m^t_1|m^t_2)$. The probability that given
$\xb_2(m^t_2)$ and $\xb_1(m^t_1|m^t_2)$, an incorrect codeword
$\xb_0$, independent of $\mb{y}^t_1$, is jointly typical with
$\mb{y}^t_1$ is asymptotically bounded by
$2^{-nI(X_0;Y_1|X_1,X_2)}$
\cite[Theorem~15.2.3]{cover_elements}. Since there are
$2^{nR_0}$ possibilities for $m_0$, the decoding at the first relay
is successful asymptotically with zero error probability if
\begin{equation}\label{eq:2relays:tx_re1:bound}
R_0\leq I(X_0;Y_1|X_1,X_2).
\end{equation}

The required constraint to ensure asymptotical zero probability
of the decoding error at the second relay is different for the two
 protocols. In the following,  each protocol is considered
 separately.\\

{\em 1) Rate Constraints for the $2^\text{nd}$ Relay in
Protocol A:} The second relay in Protocol A decodes the source
message with the help of the message of the first relay. More
specifically, in block $t$, the second relay decodes
$m^{t-1}_0$ with the help of $m^t_1=P_{\B_1}(m^{t-1}_0)$.
Assume that $m^{t-2}_0$ is successfully decoded prior to block
$t$. Knowing $m^t_2$ and $m^{t-1}_1$ in block $t$, the second
relay decodes the source message over two successive blocks by
finding a pair of messages $m_0$ and $m_1$ satisfying
$m_1=P_{\B_1}(m_0)$, such that $\xb_0(m_0|m_1,m^{t-1}_2)$ is
jointly typical with $\mb{y}_2^{t-1}$ given
$\xb_1(m_1|m^{t-1}_2)$ and $\xb_2(m^{t-1}_2)$, and
$\xb_1(m_1|m^{t}_2)$ is jointly typical with $\mb{y}_2^t$ given
$\xb_2(m^t_2)$.

The probability that an incorrect $\xb_0$ is jointly typical with
$\mb{y}^{t-1}_2$ given $\xb_1$ and $\xb_2$, is asymptotically
bounded by $2^{-nI(X_0;Y_2|X_1,X_2)}$; the probability that an
incorrect $\xb_1$ is jointly typical with $\mb{y}^t_2$ given
$\xb_2$ is asymptotically bounded by $2^{-nI(X_1;Y_2|X_2)}$. The
analysis of the probability of error closely follows the  one in
Section \ref{sec:1relay:joint_dec}. Let $\Ev_1$ denote the event
that $\xb_0$ is decoded incorrectly, and $\Ev_2$ be the event that
$\xb_1$ is decoded incorrectly. The decoding error probability is
given by $p(\Ev_1)$, which can be bounded as
$\text{Pr}(\Ev_1)=\text{Pr}\left(\Ev_1\cap(\Ev_2\cup\Ev_2^c)\right)\leq
\text{Pr}(\Ev_1\cap\Ev_2)+\text{Pr}(\Ev_1\cap\Ev_2^c)$. Now,
$\text{Pr}(\Ev_1\cap\Ev_2)$ is asymptotically bounded by
$2^{nR_0}2^{-nI(X_0;Y_2|X_1,X_2)}2^{-nI(X_1;Y_2|X_2)}$ as
there are $2^{nR_0}$ pairs of $m_0$ and $m_1$ messages in total
(since $m_1$ is a function of $m_0$). Similarly,
$\text{Pr}(\Ev_1\cap\Ev_2^c)$ is asymptotically bounded by
$2^{n(R_0-R_1)}2^{-nI(X_0;Y_2|X_1,X_2)}$, since knowing
$m_1$, there remain $2^{n(R_0-R_1)}$ choices for $m_0$. Hence,
the decoding error probability at the second relay is asymptotically
zero if
\begin{subequations}\label{eq:2relaysA:2ndre:bound}
\begin{align}
\iftwocol
R_0&\leq I(X_0;Y_2|X_1,X_2)+I(X_1;Y_2|X_2)\\
&=I(X_0,X_1;Y_2|X_2)\nonumber\\
\else
\intertext{and}
R_0&\leq I(X_0;Y_2|X_1,X_2)+I(X_1;Y_2|X_2)=I(X_0,X_1;Y_2|X_2)\\
\fi
R_0&\leq I(X_0;Y_2|X_1,X_2)+R_1
\end{align}
\end{subequations}

{\em 2) Rate Constraints for the $2^\text{nd}$  Relay in Protocol
B:} In Protocol B, the second relay only decodes the message of the
first relay, i.e., in block $t$, the second relay decodes $m^t_1$.
Decoding is performed by finding a codeword $\xb_1(m_1)$ that is
jointly typical with $\mb{y}_2^t$ given $\xb_2(m^t_2)$. The
probability that a codeword $\xb_1$, independent of $\mb{y}_2^t$,
is incorrectly decoded as the transmitted codeword by the first relay
is asymptotically bounded by $2^{-nI(X_1;Y_2|X_2)}$. Hence,
successful decoding at the second relay is possible asymptotically if
\begin{equation}\label{eq:2relaysB:2ndre:bound}
R_1\leq I(X_1;Y_2|X_2).
\end{equation}

{\em 3) Rate Constraints at the Destination:} The required rate
constraints at the destination to ensure asymptotically zero
probability of decoding error are the same for both protocols. The
decoding procedure at the destination is similar to the one at the
second relay in Protocol A. The destination decodes the source
message with the help of the messages of the first relay and the
second relay. Specifically, in block $t$, the destination decodes
$m^{t-2}_0$ with the help of $m^{t-1}_1=P_{\B_1}(m^{t-2}_0)$
and $m^{t}_2=P_{\B_2}(m^{t-1}_1)$. Decoding is performed by
jointly finding three messages $m_0$, $m_1$, and $m_2$, such that
the codeword $\xb_2(m_2)$ is jointly typical with $\yb_3^t$;
$\xb_1(m_1|m_2)$ is jointly typical with $\mb{y}_3^{t-1}$ given
$\xb_2(m_2)$; and $\xb_0(m_0|m_1,m_2)$ is jointly typical with
$\mb{y}_3^{t-2}$ given  $\xb_1(m_1|m_2)$ and $\xb_2(m_2)$.
Asymptotically for large $n$, the error probability of the joint
typicality test of $\xb_2$ and $\yb_3^t$ is  equal to
$2^{-nI(X_2;Y_3)}$; the error probability of the joint typicality test
of $\xb_1$ and $\yb_3^{t-1}$ given $\xb_2$, is
$2^{-nI(X_1;Y_3|X_2)}$; and the error probability of the joint
typicality test of $\xb_0$ and $\yb_3^{t-2}$ given $\xb_1$ and
$\xb_2$, is $2^{-nI(X_0;Y_3|X_1,X_2)}$. There are $2^{nR_0}$
valid combinations\footnote{\label{fnt:valid_m} A valid
combination corresponds to a set of values for $m_0$, $m_1$, and
$m_2$ such that $m_2=P_{\B_2}(m_1)$ and
$m_1=P_{\B_1}(m_0)$.} of $m_0$, $m_1$, and $m_2$ messages;
$2^{n(R_0-R_2)}$ combinations of $m_0$ and $m_1$ messages for
a given  $m_2$; and $2^{n(R_0-R_1)}$ choices for $m_0$ given the
two fixed bin indices $m_1$ and $m_2$ (note that $m_2$ is also
fixed when $m_1$ is fixed). Consequently, by an analysis similar to
the one for \eqref{eq:2relaysA:2ndre:bound}, the  probability of
decoding error at the destination approaches zero asymptotically if
the following constraints are satisfied:
\begin{subequations}\label{eq:2relays:dest:bound}
\begin{align}
\iftwocol
R_0&\leq I(X_0;Y_3|X_1,X_2)+I(X_1;Y_3|X_2)+I(X_2;Y_3)\\
&=I(X_0,X_1,X_2;Y_3)\nonumber\\
\else
R_0&\leq I(X_0;Y_3|X_1,X_2)+I(X_1;Y_3|X_2)+I(X_2;Y_3)=I(X_0,X_1,X_2;Y_3)\\
\fi
R_0&\leq I(X_0;Y_3|X_1,X_2)+I(X_1;Y_3|X_2)+R_2\label{eq:2relays:dest:bound:b}\\
R_0&\leq I(X_0;Y_3|X_1,X_2)+R_1 \label{eq:2relays:dest:bound:c}.
\end{align}
\end{subequations}

The following theorems summarize the achievable rates of the two
protocols for a two-relay network.

\begin{theorem}[Achievable Rate of Protocol A]
For a memoryless two-relay network defined by
$p(y_1,y_2,y_3|x_0,x_1,x_2)$, fixing any $p(x_0,x_1,x_2)$, the
source rate $R_0$ satisfying the following constraints is achievable:
\begin{subequations}\label{eq:2relays:A:rate}
\begin{align}
R_0&\leq I(X_0;Y_1|X_1,X_2)\label{eq:2relays:A:rate:a}\\
R_0&\leq I(X_0,X_1;Y_2|X_2)\label{eq:2relays:A:rate:b}\\
R_0&\leq I(X_0,X_1,X_2;Y_3)\label{eq:2relays:A:rate:c}.
\end{align}
\end{subequations}
\end{theorem}
\begin{proof}
The above rate is obtained by combining
(\ref{eq:2relays:tx_re1:bound}), (\ref{eq:2relaysA:2ndre:bound}),
and (\ref{eq:2relays:dest:bound}), and using the fact that
constraints involving $R_1$ and $R_2$ can be ignored, since $R_1$
and $R_2$ only appear on the right-hand side of all the inequalities
and can be increased freely. This rate is previously derived in
\cite{xie_kumar_rate} using a regular encoding approach.
\end{proof}

\begin{theorem}[Achievable Rate of Protocol B]
For a memoryless two-relay network defined by
$p(y_1,y_2,y_3|x_0,x_1,x_2)$, fixing any $p(x_0,x_1,x_2)$, the
source rate $R_0$ satisfying the following constraints is achievable:
\begin{subequations}\label{eq:2relays:B:rate}
\begin{align}
R_0&\leq I(X_0;Y_1|X_1,X_2)\label{eq:2relays:B:rate:a}\\
R_0&\leq I(X_0;Y_3|X_1,X_2)+I(X_1;Y_2|X_2)\label{eq:2relays:B:rate:b}\\
R_0&\leq I(X_0,X_1,X_2;Y_3)\label{eq:2relays:B:rate:c}.
\end{align}
\end{subequations}
\end{theorem}
\begin{proof}
The  above rate constraints are derived from
(\ref{eq:2relays:tx_re1:bound}),
(\ref{eq:2relaysB:2ndre:bound}), (\ref{eq:2relays:dest:bound}).
Note that (\ref{eq:2relays:B:rate:b}) is a consequence of
(\ref{eq:2relays:dest:bound:c}) and
(\ref{eq:2relaysB:2ndre:bound}). Also,
(\ref{eq:2relays:dest:bound:b}) can be ignored, since $R_2$
appears only on the right-hand side  and can be freely
increased.
\end{proof}

\subsection{Protocol A versus Protocol B}
The rate achieved by Protocol A can be higher or lower than the
rate achieved by Protocol B depending on channel parameters.
Protocol A can be capacity achieving if the network is degraded
in a particular way. Protocol B can  also be capacity
achieving, however, under a different degradedness condition.


The cut-set bound can be used to identify networks for which
Protocols A and B are capacity achieving. According to the
cut-set bound for a two-relay network, the source  rate $R_0$
 satisfies the following inequalities
 for some joint distribution $p(x_0,x_1,x_2)$ \cite[Chapter~14]{cover_elements}:
\begin{subequations}\label{eq:cut-set}
\begin{eqnarray}
R_0&\leq& I(X_0;Y_1,Y_2,Y_3|X_1,X_2)\label{eq:cut-set:1}\\
R_0&\leq& I(X_0,X_1;Y_2,Y_3|X_2)\label{eq:cut-set:2}\\
R_0&\leq& I(X_0,X_1,X_2;Y_3)\label{eq:cut-set:3}.
\end{eqnarray}
\end{subequations}

Protocol A is capacity achieving for a two-relay network in which
$X_0-(X_1,X_2,Y_1)-(Y_2,Y_3)$ and $(X_0,X_1)-(X_2,Y_2)-Y_3$
form Markov chains \cite{xie_kumar_rate}. For such a network, the
rate (\ref{eq:2relays:B:rate}) is equivalent to (\ref{eq:cut-set}),
since $I(X_0;Y_1,Y_2,Y_3|X_1,X_2)=I(X_0;Y_1|X_1,X_2)$ if
$X_0-(X_1,X_2,Y_1)-(Y_2,Y_3)$ is a Markov chain, and
$I(X_0,X_1;Y_2,Y_3|X_2)=I(X_0,X_1;Y_2|X_2)$ if
$(X_0,X_1)-(X_2,Y_2)-Y_3$ is a Markov chain.

On the other hand, Protocol B achieves the capacity of a two-relay
network in which $X_0-(X_1,X_2,Y_1)-(Y_2,Y_3)$,
$X_1-(X_2,Y_2)-Y_3$, and $X_0-(X_1,X_2,Y_3)-Y_2$ form Markov
chains. For this channel, the equivalence of (\ref{eq:cut-set:1}) and
(\ref{eq:2relays:B:rate:a}) is a direct consequence of
$X_0-(X_1,X_2,Y_1)-(Y_2,Y_3)$. It remains to show that
(\ref{eq:cut-set:2}) reduces to (\ref{eq:2relays:B:rate:b}), which
can be proved as follows: \iftwocol
\begin{align}
&I(X_0,X_1;Y_2,Y_3|X_2)\nonumber\\
&\overset{(a)}=I(X_1;Y_2,Y_3|X_2)+I(X_0;Y_2,Y_3|X_1,X_2)\nonumber\\
&\overset{(b)}=I(X_1;Y_2|X_2)+I(X_0;Y_2,Y_3|X_1,X_2)\nonumber\\
&\overset{(c)}{=}I(X_1;Y_2|X_2)+I(X_0;Y_3|X_1,X_2)\nonumber
\end{align}
\else
\begin{align}
I(X_0,X_1;Y_2,Y_3|X_2)&\overset{(a)}=I(X_1;Y_2,Y_3|X_2)+I(X_0;Y_2,Y_3|X_1,X_2)\nonumber\\
&\overset{(b)}=I(X_1;Y_2|X_2)+I(X_0;Y_2,Y_3|X_1,X_2)\nonumber\\
&\overset{(c)}{=}I(X_1;Y_2|X_2)+I(X_0;Y_3|X_1,X_2)\nonumber
\end{align}
\fi where (a) follows from the chain rule for the mutual information,
(b) follows from $X_1-(X_2,Y_2)-Y_3$, and (c) follows from
$X_0-(X_1,X_2,Y_3)-Y_2$.  For future reference, we call this type of
degraded network {\em doubly degraded}.

Doubly degraded network corresponds to a network in which the
channel from the source to the second relay is blocked. The
interpretation of the first Markov chain
$X_0-(X_1,X_2,Y_1)-(Y_2,Y_3)$ is that the channel from the source
to the first relay is stronger than the channel from the source to the
second relay and the destination. The second Markov chain
$X_1-(X_2,Y_2)-Y_3$ states that  the channel from the first relay to
the second relay  is stronger than the channel from the first relay to
the destination. The last Markov chain $X_0-(X_1,X_2,Y_3)-Y_2$
implies that the channel from the source to the second relay is
weaker than the channel from the source to the destination. The
next example describes a
Gaussian version of such a doubly degraded network.\\



\EX Consider the two-relay network depicted in Fig.\
\ref{fig:cross_degraded}. In this network, the source signal is
represented by $X_0$. The first relay receives $Y_1=X_0+N_1$
where $N_1\sim\mathcal{N}(0,\sigma_1^2)$  and transmits $X_1$.
The channel from the source to the second relay is blocked. The
second relay receives $Y_2=X_1+N_2$, and transmits $X_2$, where
$N_2\sim\mathcal{N}(0,\sigma^2_2)$. Destination receives
$Y_3=X_0+X_1+X_2+N_1+N_2+N_3$, where
$N_3\sim\mathcal{N}(0,\sigma_3^2)$.

It can be shown that a joint Gaussian distribution is the optimal
input distribution for such a degraded additive Gaussian noise
network. This is proved for the degraded single-relay channel with
additive Gaussian noise in \cite[Section IV]{cover_elgamal}. The
same technique is applicable to  the multirelay case. We skip the
proof and assume the optimality of the jointly Gaussian input in the
following.

The achievable rate of Protocol B is derived by defining the joint
distribution of $(X_0,X_1,X_2)$ as follows. Let
$X_0=X'_0+\alpha_{11}X'_1+\alpha_{12}X'_2$ where
$X'_0\sim\mathcal{N}(0,Q_0)$, $X'_1\sim\mathcal{N}(0,Q_1)$, and
$X'_2\sim\mathcal{N}(0,Q_2)$ are independent Gaussian random
variables. For the first relay, let $X_1=X'_1+\alpha_{22}X'_2$, and
for the second relay, set $X_2=X_2'$. The correlations among the
three input signals are controlled by $\alpha_{ij}, 1\leq i\leq j\leq
2$.

\begin{figure}
\begin{center}
\iftwocol
\includegraphics[scale=0.7]{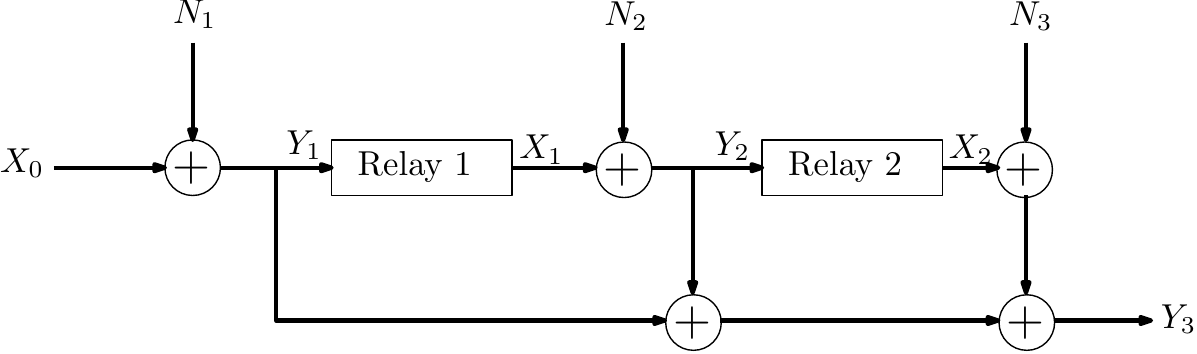}
\else
\includegraphics[scale=1]{cross_degraded.pdf}
\fi
\caption{A Gaussian doubly degraded
two-relay network.\label{fig:cross_degraded}}
\end{center}
\end{figure}

The network is doubly degraded, since the source can only
communicate to the destination through the first relay, the first
relay can only communicate to the destination through the second
relay, and  the source cannot  directly communicate to the second
relay.



The achievable rate of Protocol A under power constraints
$\mathbb{E}{X^2_i}\leq P_i, 0\leq i\leq 2$ is given by
\begin{subequations}\label{eq:2r_gaussian_rate}
\begin{align}
R_0&< \frac{1}{2}\log_2\left(1+\frac{Q_0}{\sigma^2_1}\right)\\
R_0&<\frac{1}{2}\log_2\left(1+\frac{Q_1}{\sigma^2_2}\right)\\
R_0&<
\frac{1}{2}\log_2\left(1+\frac{Q_0+(1+\alpha_{11})^2Q_1+(1+\alpha_{12}+\alpha_{22})^2Q_2}{\sigma^2_1+\sigma^2_2+\sigma^2_3}\right).
\end{align}
\end{subequations}
maximized over $\alpha_{ij}, 1\leq i\leq j\leq 2$, $Q_0$, $Q_1$,
and $Q_2$ such that
$Q_0+\alpha_{11}^2Q_1+\alpha_{12}^2Q_2\leq
P_0,Q_1+\alpha_{22}^2Q_2\leq P_1$, and $Q_2\leq P_2$. This rate
is also equal to the multihop rate of \cite{xie_kumar_rate} for this
network.

The  power-constrained achievable rate of Protocol B given in
(\ref{eq:2relays:B:rate}) translates into the following constraints:
\begin{subequations}
\begin{align}
R_0&< \frac{1}{2}\log_2\left(1+\frac{Q_0}{\sigma^2_1}\right)\nonumber\\
R_0&<\frac{1}{2}\log_2\left(1+\frac{Q_0}{\sigma^2_1+\sigma^2_2+\sigma^2_3}\right)+\frac{1}{2}\log_2\left(1+\frac{Q_1}{\sigma^2_2}\right)\nonumber\\
R_0&<
\frac{1}{2}\log_2\left(1+\frac{Q_0+(1+\alpha_{11})^2Q_1+(1+\alpha_{12}+\alpha_{22})^2Q_2}{\sigma^2_1+\sigma^2_2+\sigma^2_3}\right)\nonumber,
\end{align}
\end{subequations}
maximized over $\alpha_{ij}, 1\leq i\leq j\leq 2$, $Q_0$, $Q_1$,
and $Q_2$ such that
$Q_0+\alpha_{11}^2Q_1+\alpha_{12}^2Q_2\leq
P_0,Q_1+\alpha_{22}^2Q_2\leq P_1$, and $Q_2\leq P_2$. This
rate, which is also the capacity of this doubly degraded network, is
strictly greater than the multihop rate (\ref{eq:2r_gaussian_rate}).
As a comparison of Protocols A and B, note that in the case that the
optimal values of $Q_0, Q_1,\sigma^2_1$, and $\sigma^2_2$ are
such that $Q_1/\sigma^2_2<Q_0/\sigma^2_1$ (e.g.,
$\sigma^2_1\ll\sigma_2^2$), it is better to turn the second relay
off rather than forcing it to decode the source message in Protocol
A; however, the second relay can still help the overall transmission
 using Protocol B. \EEX

\section{ Multirelay Parity Forwarding \label{sec:multirelay}}
The parity forwarding protocol can be generalized to multirelay
networks by allowing relay terminals   to transmit
 messages that are bin indices of the messages of the source or other relays.
The first step in such a structured generalization is to specify the
relation among messages.

The relation among messages in a parity forwarding protocol is
transitive and can take various forms. For example, a relay
message A can be a random bin index of another relay message B,
while message B itself is a bin index of another message C.
Thus, message C is also a bin index for message A. Another
possibility could be that messages B and C are independent
random bin indices of message A. In this paper, we propose the
use of a message tree to describe the relation among messages.

Consider an example of a  parity forwarding protocol in a network
with four relays shown in Fig.~\ref{fig:4relays_net}.
Fig.~\ref{fig:message_sets} shows a message tree that specifies
one possible relationship between messages and their parities.   In
the message tree, a node represents  a message. The root node
represents the source message $m_0$. Branching out from a
message  corresponds to forming a parity for that message, i.e.,
over each branch,    the child message is a parity (random bin
index) for the parent message.

There are several possibilities for associating messages in the tree
with transmitting nodes. The message sets
$\A_0,\A_1,\A_2,\A_3,\A_4$ in Fig.~\ref{fig:4relays_net} describe
one such association for the message tree in
Fig.~\ref{fig:message_sets}. In general, the message sets partition
the set of all messages in the tree and each partition is associated
with a transmitter (source or relay). Several partition schemes may
exist for a message tree, resulting in   different parity forwarding
protocols with different rates.

A relay node computes the messages associated with it by
decoding a set of messages from other nodes. For example, all
parity messages are computable if the relay decodes the source
message $m_0$. However, decoding the source message is not the
only possibility. For example, the second relay in
Fig.~\ref{fig:message_sets}   only needs
  to decode the message $m_{11}$ in order to compute
$\A_2=\{m_{21}\}$. The  set of messages that a relay decodes is
called the decoding message set. For the second relay in this
example, the decoding set can be  $\D_2=\{m_0\}$, or
$\D_2=\{m_{11}\}$, or $\D_2=\{m_{11},m_{12}\}$, because  the
message set of the second relay $\A_2=\{m_{21}\}$ is directly a
parity of $m_{11}$ and indirectly a parity of $m_0$. For the
destination, the decoding set is always the set $\{m_0\}$.
Different choices of the decoding sets result in different rates
corresponding to different parity forwarding protocols. For example,
for the network with two relays in the previous section, different DF
rates are achieved depending on which message the second relay
decodes.

In summary, a parity forwarding protocol for a multirelay network is
characterized by three components:
\begin{itemize}
\item A message tree which specifies a set of messages and
    their     relation with respect to each other;
\item The message sets which    associate messages in the
    message tree with the transmitters     (source or the
    relays); and
 \item  The decoding sets which specify the messages that the
     relays decode.
 \end{itemize}
\begin{figure}
\centering
\iftwocol
\includegraphics[scale=0.6]{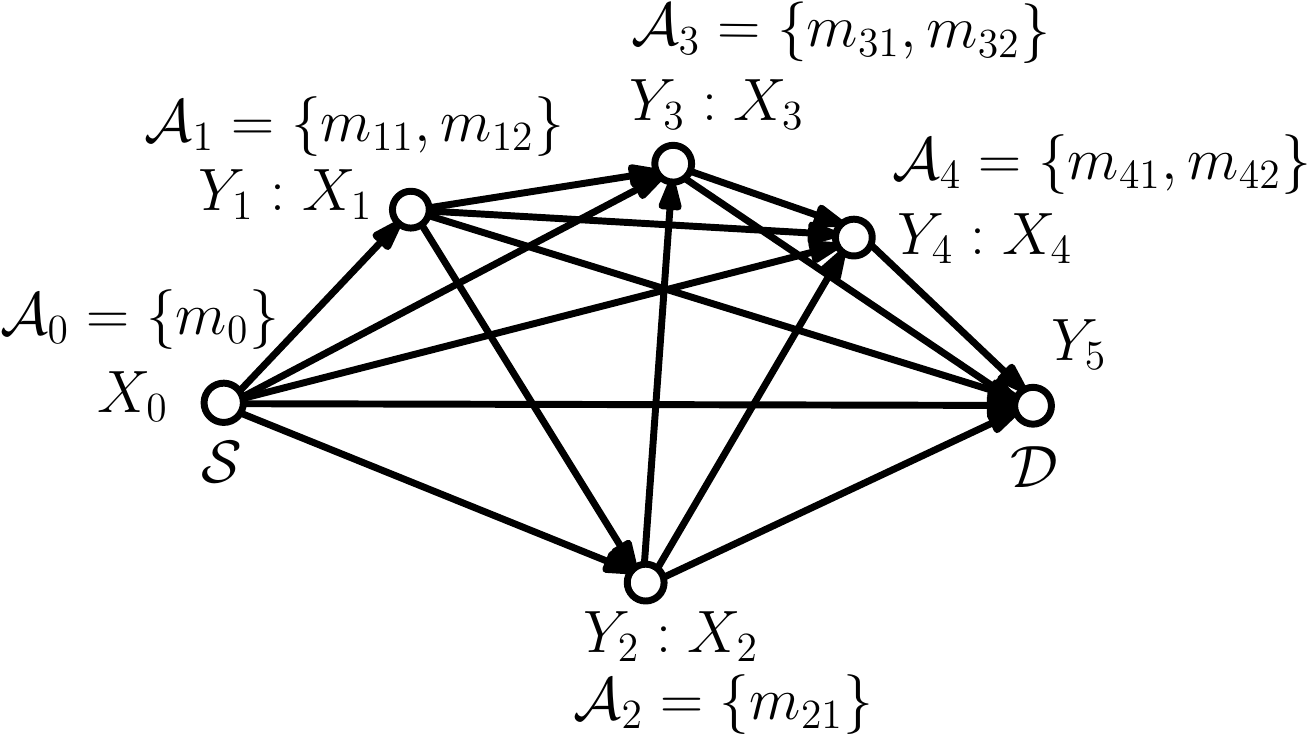}
\else
\includegraphics[scale=0.7]{net_4relays.pdf}
\fi
\caption{A network with four relays. The message sets $\A_0$ up to
$\A_4$ specify the messages each relay encodes.
\label{fig:4relays_net}.}
\end{figure}

\begin{figure}
\centering
\iftwocol
\includegraphics[scale=0.5]{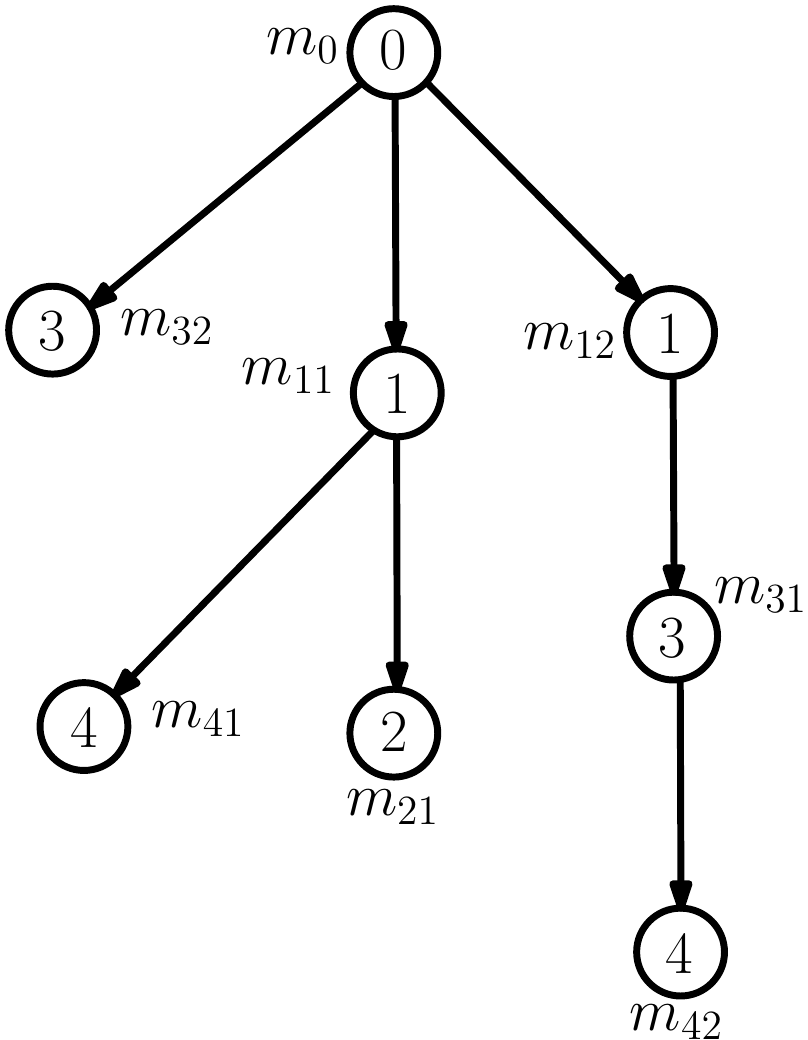}
\else
\includegraphics[scale=0.7]{msg_tree_4relays.pdf}
\fi
\caption{An example of message tree and message sets for a
network with four relays shown in
Fig.~\ref{fig:4relays_net}.\label{fig:message_sets}}
\end{figure}

\subsection{Mathematical Formulation}
Formally, consider a general multirelay network consisting of a pair
of  source and destination and $K$ relays  numbered from 1 to $K$.
Let the set $\A_0=\{m_{01},m_{02},\ldots,m_{0n_0}\}$ denote
the set of messages encoded by the source, where $m_{0n_0}$ is
the source message and the rest of messages are parities for
$m_{0n_0}$. For notational simplicity, let the subscript ``$0n_0$''
be equivalent to ``$0$''  so that, for example, $m_0=m_{0n_0}$.
In each block, the $k$th relay, $1\leq k\leq K$, transmits a set of
$n_k>0$ parity messages $\A_k=\{m_{k1}$, $m_{k2}$,$\ldots$ ,
$m_{kn_k}\}$. The message tree, defined in the following,
specifies the relation among the messages.

\begin{definition}[Message Tree]
The message tree is defined by a directed tree
${\tau}=(\M,\mathcal{V})$ where $\M$ denotes the set of nodes
and $\mathcal{V}$ denotes the set of directed edges. Each node
in the message tree represents a message. The source message
$m_0$ is associated with the root node. All other nodes
correspond to parity messages sent by the relay terminals.  An
edge corresponds to forming a bin index. Branching out from the
message $m_{k'l'}$ to the message $m_{kl}$ corresponds to
forming the bin index $m_{kl}$ for the message $m_{k'l'}$ with
respect to a random  partition $\B_{kl}$ of the message space
$\{1,2,\cdots,2^{nR_{k'l'}}\}$ into $2^{nR_{kl}}$ bins, i.e.,
$m_{kl}=P_{\B_{kl}}(m_{k'l'})$. The random partition sets
$\B_{kl}$, $1\leq k\leq K, 1\leq l\leq n_k$ are generated
independently.\end{definition}

\begin{notation}[$\ra$, $\nra$]
We use $m\ra m'$ to denote that the message $m'$ is a descendent
of the message $m$. The message $m'$ is a descendent of  $m$ if
there is a path from the message $m$ to the message $m'$ in the
message tree, i.e., $m'$ is a parity message for $m$, either directly
or indirectly via other intermediate parity messages. A message is
considered to be a descendent (parity) of itself, i.e., $\forall
m\in\M, m\ra m$. We write $m\nra m'$ if $m\ra m'$ is not true. We
write $\mathcal{F}\ra m$ to denote that the message $m$ is a
descendent of the set $\mathcal{F}\subset \M$. The message $m$
is a descendent of ${\cal F}$ if $\exists f\in\mathcal{F}$ such that
$f\ra m$. We write $\mathcal{F}\ra\mathcal{G}$ to denote that the
set $\mathcal{G}\subset \M$ is a descendent of $\mathcal{F}$. The
set $\mathcal{G}$ is a descendent of $\mathcal{F}$ if
$\mathcal{F}\ra g, \forall g\in \mathcal{G}$; in this case,
$\mathcal{F}$ is called a generator for $\mathcal{G}$. The
``$\nra$'' relation for sets of messages denotes the negation of
``$\ra$''.
\end{notation}

The collection of the message sets $\A_k$'s forms a disjoint
partition of $\M$. The $\A_k$ sets are such that for $i<j$, there are
no messages in $\A_i$ that are a descendent of $\A_j$. This is
required since messages in $\A_k$ are constructed by forming
parities for messages in $\A_0$, $\A_1$, $\ldots$, $\A_{k-1}$. A
message is said to be of order $k$ if it is transmitted by the $k$th
relay. Note that a message and its parity may be transmitted by the
same relay; this allows for partial decode-and-forward schemes of
the type introduced in \cite[Theorem~7]{cover_elgamal} (see the
example in Section \ref{sec:ex:aref}).

Having specified the relation of messages with each other, we now
need to specify in which time slot a message is transmitted. In
block $t$, the $k$th relay encodes an instance of $\A_k$ denoted by
$\A^t_k=\{m^t_{k1},m^t_{k2},\cdots,m^t_{kn_k}\}$. The
messages sent in block $t$ are parities of the messages sent in
blocks $t-k$ to $t-1$. More precisely, if $m_{ki}$ is a child (direct
parity) of $m_{lj}$ in the message tree, then in block $t$ we have:
\begin{equation}\label{eq:timing}
m^t_{ki}=P_{\B_{ki}}(m^{t-(k-l)}_{lj}).
\end{equation}

The set of all messages decoded by the $k$th relay is specified by
the set $\D_k$. The message set $\A_k$ is a descendent of $\D_k$.
By decoding the messages in $\D_k$, the $k$th relay knows all
messages that are directly or indirectly parities of messages in
$\D_k$, some of which are assigned to other relay terminals to be
transmitted in subsequent blocks.  For optimal encoding, the $k$th
relay should cooperate with other relay terminals to transmit
messages known to the $k$th relay. The $k$th relay utilizes
superposition broadcast encoding to encode multiple messages,
which imposes limitations on the choice of $\D_k$ sets.

The $\D_k$ sets should have the following properties. First,
the set $\A_k$ can be generated from $\D_k$,  i.e.,
$\D_k\ra\A_k$, while $\D_k\cap\A_k=\{\}$. Second, a message and
its parity cannot be both in $\D_k$, since the parity can be
computed and thus removed from $\D_k$. Third, in order for the
sets $\D_k$ to be consistent with the encoding method described
in the next section, the $\D_k$ sets should be such that if a
message $m_{li}$ belongs to the set $\D_k$,  then all messages
of the same order $l$ and with a smaller second subscript
should be a descendent of $\D_k$, i.e. $\forall i'\leq
i:\D_k\ra m_{li'}$ (note that if $m_{rs}\in \D_k$ then $\D_k\ra
m_{rs}$). For the destination, we have $\D_{K+1}=\{m_0\}$. It
should be noted that a collection of $\D_k$ sets satisfying
these features always exists, since $\D_k=\{m_0\}$, $1\leq
k\leq K+1$, satisfies the above properties. (If $\D_k=\{m_0\}$,
$1\leq k\leq K+1$, then all the relays decode the source
message. This choice of decoding sets gives the multihop DF
rate, irrespective of the underlying message tree.)

\begin{figure}
\centering
\iftwocol
\includegraphics[scale=0.5]{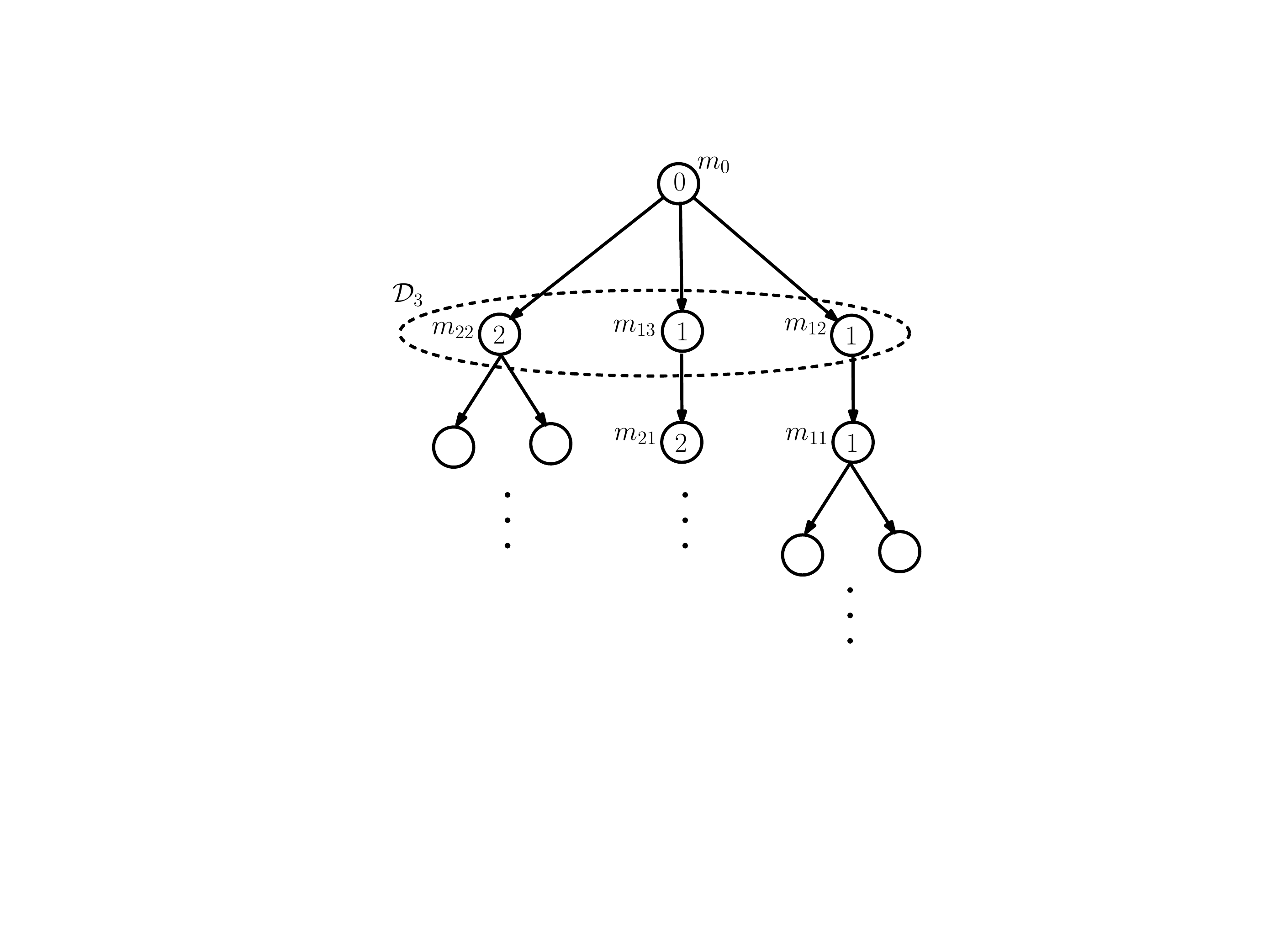}
\else
\includegraphics[scale=0.7]{why_D.pdf}
\fi
\caption{Illustration of the properties of decoding sets.\label{fig:why_D}}
\end{figure}

To elaborate on the properties of $\D_k$ sets, consider the
message tree in Fig.~\ref{fig:why_D}. In this example,
$\D_3=\{m_{12},m_{13},m_{22}\}$ is consistent with all
properties specified above. However,
$\D_3=\{m_{11},m_{13},m_{22}\}$ is invalid, because it violates
the third property. The reason is that
 $m_{13}$ is encoded on top of $m_{11}$ and
$m_{12}$ in the superposition encoding procedure described in the
next section. Hence, decoding $m_{13}$ requires the decoding of
$m_{11}$ and in particular $m_{12}$. But, $m_{12}\notin \D_3$
and $\D_3\nra m_{12}$. Thus, $\D_3=\{m_{11},m_{13},m_{22}\}$
does not completely specify all the messages decoded, violating the
definition of the decoding set.

In summary, a parity forwarding protocol for a network with $K$
relays is defined by a three tuple $(\mathcal{\tau},\A,\D)$
consisting of the message tree $\mathcal{\tau}$, defining the
relation between messages and parities; a partition
$\A=\{\A_0,\A_1,\ldots,\A_{K}\}$ which assigns the messages to
different relay terminals; and the set of decoding sets
$\D=\{\D_1,\D_2,\ldots,\D_{K+1}\}$, which determines the set of
messages each relay should decode.


\subsection{Encoding\label{sec:multirelay:encoding}}
Messages are encoded using superposition coding. Consider as an
example, the encoding of three messages
$m_a\in\{1,2,\ldots,2^{nR_a}\}$,
$m_b\in\{1,2,\ldots,2^{nR_b}\}$,
$m_c\in\{1,2,\ldots,2^{nR_c}\}$. First, $2^{nR_c}$ codewords
$\xb_c(m_c)$ are randomly generated according to  $p(x_c)$. Next,
$2^{nR_b}$ codewords $\xb_b(m_b|m_c)$ are generated for each
$\xb_c$ codeword according to  $p(x_b|x_c)$. We call the set
$\{m_c\}$ the known message set and the set
$\{\xb_c(m_c)|m_c=1,2,\ldots,2^{nR_c}\}$ the known codeword set
for $m_b$, as the $\xb_b$ codewords are generated for fixed known
values of $\xb_c(m_c)$. Similarly,  $\xb_a$ codewords are randomly
generated for every pair of $m_b$ and $m_c$ messages by randomly
choosing a codeword $\xb_c(m_a|m_b,m_c)$ according to
$p(x_a|x_b,x_c)$ conditioned upon fixed $\xb_b(m_b|m_c)$ and
$\xb_c(m_c)$ codewords. Accordingly, the known message sets and
codeword sets for $m_a$ are  $\{m_b,m_c\}$ and
$\{\left(\xb_b(m_b|m_c\right),\xb_c(m_c))|m_b=1,\ldots,2^{nR_b},m_c=1,\ldots,2^{nR_c}\}$,
respectively. The codeword $\xb_a(m_a|m_b,m_c)$ encodes the
three messages $m_a$, $m_b$, and $m_c$.

In the parity forwarding protocol, messages are encoded following
the same procedure as in the above example using superposition
coding. The key is to identify the known message sets for each
message $m_{kl}$ in the message tree to determine the set of all
messages on top of which $m_{kl}$ should be encoded.

\begin{definition}[Known Sets]\label{def:known_sets}
The message $m_{ki}$ is superposed onto its known message sets
$\C^m_{ki}$. First, within each $\A_k$, messages are superposed
onto each other in the order of their second subscripts, i.e., at the
$k$th relay terminal, the message $m_{ki}$ is encoded on top of
all messages $m_{kj},j<i$. Furthermore, in each block, the $k$th
relay knows all messages of  an order higher than $k$ which are
descendants of $\D_k$, i.e., $\{m_{k'i'}\in \M|k'>k,\D_k\ra
m_{k'i'}\}$.  Hence, the {\em known message set} of $m_{ki}$ is
$\C^m_{ki}=\{m_{kj},j<i\}\cup\{m_{k'i'}\in \M|k'>k,\D_k\ra
m_{k'l'}\}$. Further, we use $\C^m_{ki}(t)$ to denote the instance
of $\C^m_{ki}$ messages in the $t$th block.
\end{definition}

As an example of known sets consider Fig.~\ref{fig:known_sets}
corresponding to the four-relay network in
Fig.~\ref{fig:4relays_net}. Let $\D_2=\{m_{11},m_{12}\}$ be the
decoding set for the message $m_{21}$. Then, the known set
$\C^{m}_{21}$ for $m_{12}$ consists of all descendants of $\D_2$
with orders higher than 2, the order of $m_{21}$.

The next step in the generation of random codebooks is to assign a
probability distribution and a random variable to each message in
the tree. Let $X_{ki}$ represent the  random variable corresponding
to the encoding of $m_{ki}$.  The set $\C^x_{ki}$ is defined as
the set of random variables corresponding to messages in
$\C^m_{ki}$, i.e., $\C^x_{ki}=\{X_{ki}|m_{ki}\in\C^m_{ki}\}$.
Let $p(x_{ki}|\C^x_{ki})$ be the conditional probability
distribution associated with $X_{ki}$. Note that by properties of
the decoding sets and known sets,
$\prod_{k=0}^K\prod_{i=1}^{n_k}p(x_{ki}|\C^x_{ki})$ is a valid
joint probability distribution.

\iftwocol
\begin{figure}
  \begin{center}
\iftwocol
      \scalebox{0.5}{\includegraphics{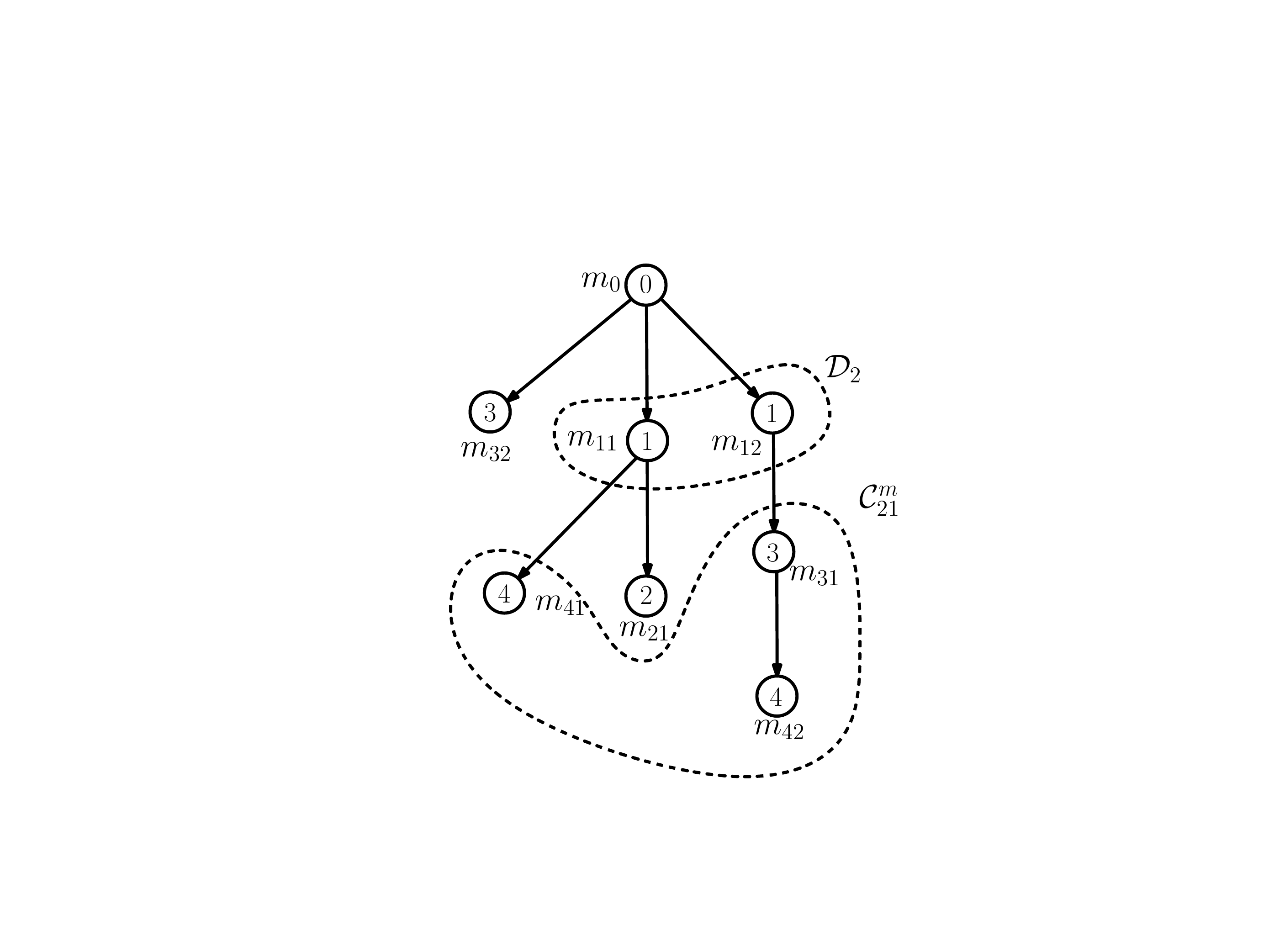}}
      \else
      \scalebox{1}{\includegraphics{k_sets6.pdf}}
      \fi
    \caption{\label{fig:known_sets} Examples of known sets for the four-relay network in Fig.~\ref{fig:message_sets}.}
  \end{center}
\end{figure}

Random codebook construction starts with messages of the $K$th
relay that have empty known sets. For every message $m_{ki},1\leq
k\leq K, 1\leq i\leq n_i$ with $\C^m_{ki}=\{\}$, $2^{nR_{ki}}$
codewords $\mb{x}_{ki}(m_{ki})$ are randomly generated
according to the probability distribution $p(x_{ki})$. In the next
step, for every message $m_{k'i'}$ for which the codewords for all
messages in the corresponding $\C^m_{k'i'}$ have already been
constructed in previous steps, $2^{nR_{k'i'}}$ random codewords
$\xb_{k'i'}(m_{k'i'}|C^m_{k'i'})$ are generated for every
combination of codewords in $\C^m_{k'i'}$ according to
$p(x_{k'i'}|\C^x_{k'i'})$. This procedure is repeated until random
codebooks are generated for all messages. In block $t$, the $k$th
terminal,  $0\leq k\leq K$, transmits
$\mb{x}_{kn_k}\left(m^t_{kn_k}|\C^m_{kn_k}(t)\right)$. This is
equivalent to having $X_{kn_k}$ as the channel input by the $k$th
relay. For notational simplicity, let $X_k\triangleq X_{kn_k},0\leq
k\leq K,$ so that the channel input by the $k$th relay is
represented by $X_k$.

\subsection{Decoding}
Messages are decoded at each relay node and at the destination via
joint decoding. A set of messages and their parity messages are
jointly decoded by finding a combination of messages consistent
with the parity relationship between messages, such that the
corresponding codewords are jointly typical with the respective
received sequences.

At the $k$th decoding node (a relay or the destination), decoding is
performed over a window of successive blocks. To identify the
decoding window of received sequences at the $k$th relay, let $q$
be the smallest order of the relays of which a message is decoded
by the $k$th relay, i.e., $q$ is the smallest number such that
$\exists l: m_{ql}\in \D_k$. Let
 $\mb{y}^t_1, \mb{y}^t_2,
\ldots,\mb{y}^t_{K+1}$ denote the random sequences representing
the  received sequences in block $t$ at the first relay, the second
relay,  up to the destination, respectively. Then, according to
(\ref{eq:timing}), the decoding window for the $k$th relay in block
$t$ is given by $(\mb{y}^{t-(k-q)+1}_k,\ldots,\mb{y}^{t}_k)$.

To identify all parity messages for the messages in $\D_k$
available to the $k$th relay terminal, note that any descendent
message of $\D_k$ sent in or before block $t$ can be used in block
$t$ as a parity for the messages in $\D_k$. Note that according to
(\ref{eq:timing}), a parity message $m_{rs}$ for $\D_k$ with
$r\geq k$ is available only after block $t$. Hence, the set of all
such messages is given by
\begin{equation}
\T_k=\{m_{rs}\in \M|\D_k\ra m_{rs}, r<k\}.
\end{equation}
See Fig.~\ref{fig:rate_describe} for an example of $\T_k$ for
$k=4$.

To decode the messages in $\T_k$, note that for each message
$m_{lj}\in\T_k$, the probability that the corresponding codeword
$\mb{x}_{lj}\left(m_{lj}|\C^m_{lj}\left(t-(k-l)+1\right)\right)$,
$q\leq l<k$, generated according to
$p\left(x_{lj}|\C^x_{lj}\left(t-(k-l)+1\right)\right)$, is incorrectly
declared jointly typical with $\mb{y}_k^{t-(k-l)+1}$,
 given $\C^\xb_{lj}\left(t-(k-l)+1\right)$, is asymptotically bounded by
$2^{-nI(X_{lj};Y_k|\C^x_{lj})}$, where
$\C^\xb_{lj}\left(t-(k-l)+1\right)$ denotes the set of
 codewords  corresponding to messages in
$\C^m_{lj}\left(t-(k-1)+1\right)$.

We now upper bound the  error probability for joint decoding of all
messages in $\T_k$. The error probability of jointly decoding all
messages in $\T_k$ approaches zero as $n$ goes to infinity if for
every subset $\I'\subset \T_k$, the probability that all messages in
$\I'$ are decoded incorrectly asymptotically approaches zero. This is
similar to bounding the probability of decoding error in the multiple
access channel. See \cite[Theorem 14.3.5]{cover_elgamal}.

To enumerate all subsets $\I'$ of incorrectly decoded messages, we
take the following approach. Let $\I$ denote any subset of $\T_k$.
Note that if messages in $\I$ are decoded correctly, then any
message in $\T_k$ which is a parity of a message in $\I$ is also
decoded correctly. Thus, the set $\I'$ must be of the form
\begin{equation}
\I'=\{m_{li}\in\T_k|\I\nra m_{li}, \I\subset \T_k\}\nonumber.
\end{equation}
In order to bound the error probability, we need to count the
number of all valid message combinations\footnote{ A valid
message combination for a set of messages corresponds to a set of
values for the messages in the set consistent with the relationships
defined by the message tree. See Footnote \ref{fnt:valid_m} for an
example.}    for the sets $\T_k$, $\I$, and $\I'$. Let
$|\mathcal{W}|$ denote the number of valid combinations of
messages in the set $\mathcal{W}$. Note that for any set
$\J_\mathcal{W}$ that is both a generator and a subset of
$\mathcal{W}$, we have $|\J_{\mathcal{W}}|=|\mathcal{W}|$,
since messages in $\mathcal{W}$ are functions of the messages in
$\J_{\mathcal{W}}$. For example, $|\T_k|=|\D_k|$ as $\D_k$ is a
subset that  generates $\T_k$. To further simplify the computation
of $|\I|$, let $\J_{\I}$ be the minimal generator of $\I$, i.e.,
$\J_{\I}\ra\I$, and $\forall \mathcal{F}\subset\I,\mathcal{F}\ra \I:
\J_{\I}\subset\mathcal{F}$.  With this definition of $\J_{\I}$, we
have $|\I|=|\J_{\I}|$. See Fig.~\ref{fig:cardinality} for an example.

To compute $|\I'|$, note that $|\I'|=|\T_k|/|\I|$. This is
because, fixing $\I$, valid combinations of messages in $\T_k$
are constrained by the bin indices in $\I$. There are $|\I|$
valid bin indices in $\I$. The number of remaining valid
message combinations, which corresponds to $|\I'|$, is
therefore $|\T_k|/|\I|$. Now, using the decoding error
probability for individual messages in $\I'$, the probability
that no message in $\I'$ is decoded correctly is given by
\begin{align}\label{eq:rateI}
&|\I'|2^{-n\sum_{\forall hj:m_{hj}\in\I'}I(X_{hj};Y_k|\C^x_{hj})}\nonumber\\
&=\frac{|\T_k|}{|\I|}2^{-n\sum_{\forall hj:m_{hj}\in\I'}I(X_{hj};Y_k|\C^x_{hj})}\nonumber \\
&=\frac{|\T_k|}{|\J_{\I}|}2^{-n\sum_{\forall hj:m_{hj}\in\I'}I(X_{hj};Y_k|\C^x_{hj})}\nonumber \\
&=\frac{|\D_k|}{|\J_{\I}|}2^{-n\sum_{\forall hj:m_{hj}\in\I'}I(X_{hj};Y_k|\C^x_{hj})}.
\end{align}
Consequently, at the $k$th terminal, $1\leq k\leq K+1$,  the
joint decoding error probability asymptotically approaches zero
if for every subset $\I$, \eqref{eq:rateI} approaches zero as
$n$ goes to infinity. This can be ensured if the following
holds:
\begin{align}\label{eq:rateT}
&\underbrace{\sum_{\forall h,i: m_{hi}\in \D_k} R_{hi}}_{\log |\D_k|}\leq\nonumber\\
&\sum_{\forall
hj:m_{hj}\in\I'}I(X_{hj};Y_k|\C^x_{hj})&+\underbrace{\sum_{\forall l,j: m_{lj}\in \J_{\I}}R_{lj}}_{\log |\J_{\I}|}\qquad\forall \I\subset \T_k.
\end{align}

\begin{figure}
\centering
\iftwocol
\includegraphics[scale=0.45]{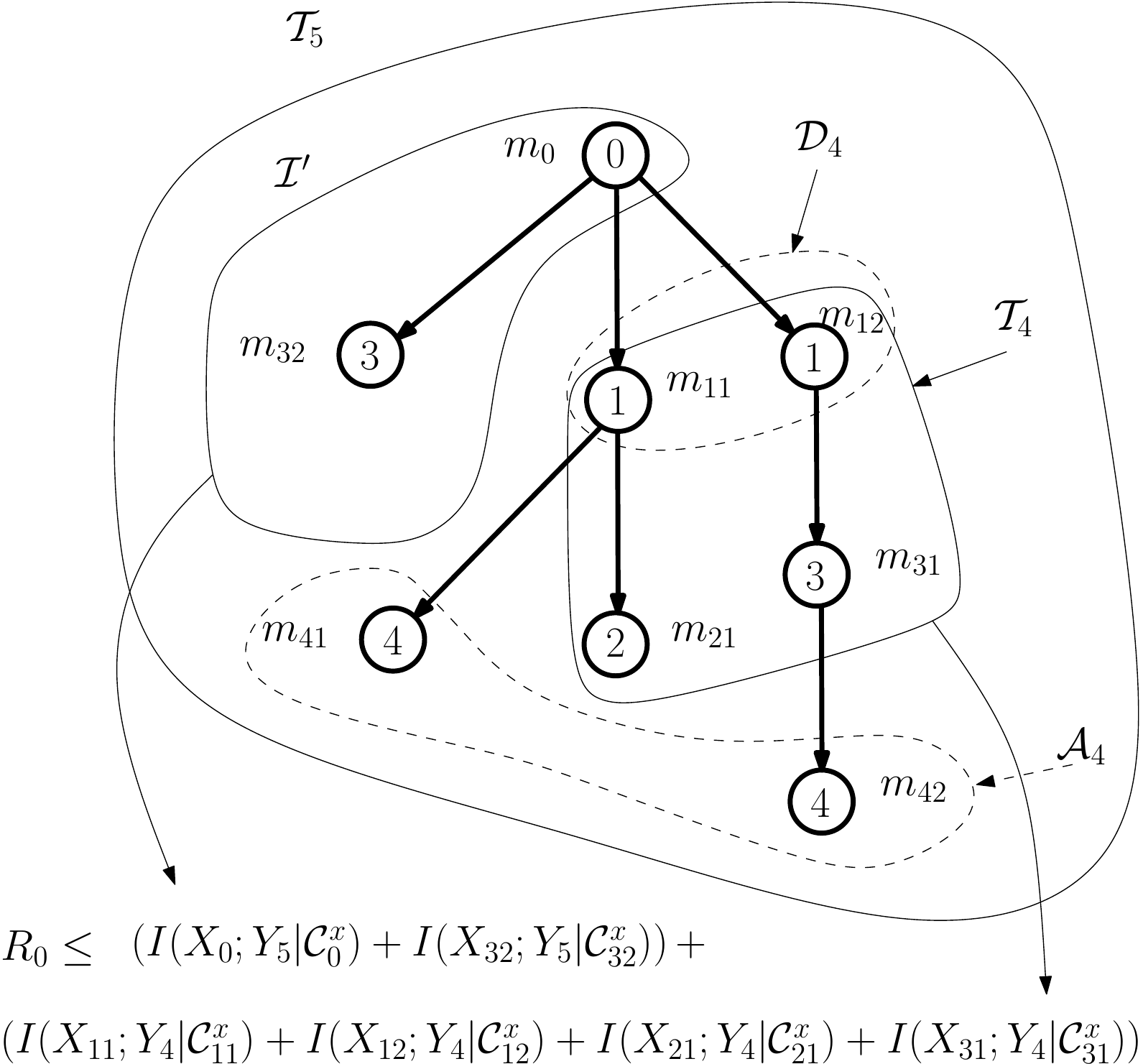}
\else
\includegraphics[scale=0.7]{rate_describe.pdf}
\fi
\caption{\label{fig:rate_describe} An example of $\D_4$, $\T_4$, and $\A_4$ sets and their effect on the rate constraints. The rate constrain is derived by combining the constraints for $\T_4$ and $\T_5$.}
\end{figure}

\begin{figure}
\centering
\iftwocol
\includegraphics[scale=0.5]{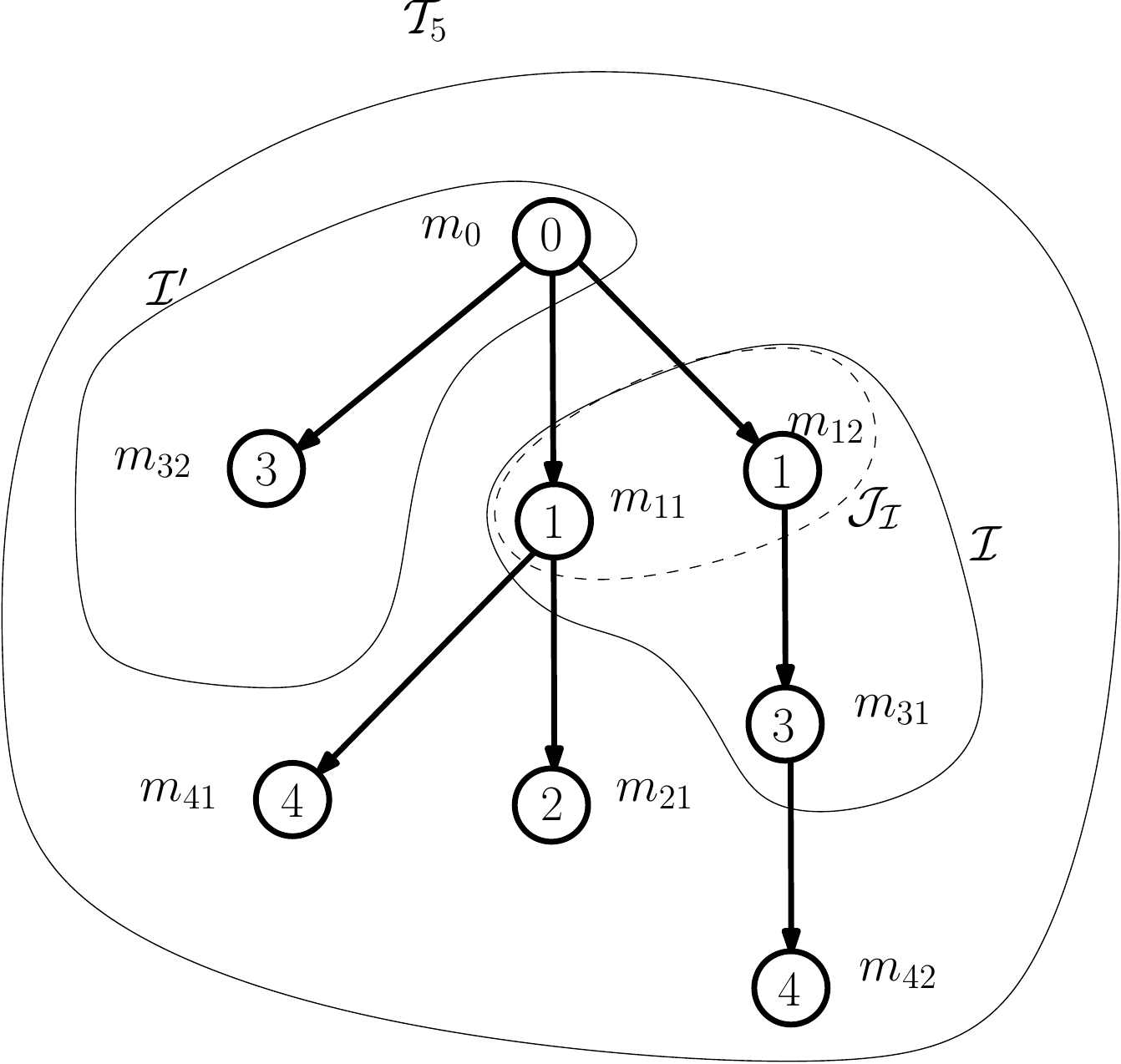}
\else
\includegraphics[scale=0.7]{cardinality.pdf}
\fi \caption{\label{fig:cardinality} An example of $\I$, $\I'$, and
$\J_{\I}$ subsets of $\T_5$, the set of all messages decoded at
the destination. Fixing messages in $\I$, the number of
valid combinations of messages in $\I'$, $|\I'|$, is given by
$|\I'|=|\T_5|/|\I|$.}
\end{figure}

Summarizing, the following theorem characterizes the achievable rate
of the parity forwarding protocol defined by $(\tau,\A,\D)$.

\begin{theorem}[Achievable Rate of Parity Forwarding]\label{thm:total_rate}
Consider a memoryless relay network with $K$ relay terminals
defined by the probability distribution\footnote{Recall that
$X_k\triangleq X_{kn_k}, 0\leq k\leq K$.}
\[p(y_1,y_2,\cdots,y_K,y_{K+1}|x_0,x_{1},x_{2},\cdots,x_{K}).\]
Using the parity forwarding protocol defined by $(\tau,\A,\D)$, the
source rate $R_{0}$ satisfying the following constraints maximized
over the probability distribution $\prod_{\forall l,i: m_{li}\in\M}
p(x_{li}|\C^x_{li})$ (along with a set of positive rates $R_{hi}$) is
achievable: \begin{align}\label{eq:rate} \sum_{\forall h,i:
m_{hi}\in \D_k} R_{hi}&\leq \sum_{\forall
hj:m_{hj}\in\I'}I(X_{hj};Y_k|\C^x_{hj})+\sum_{\forall l,j: m_{lj}\in \J_{\I}}R_{lj}\nonumber\\
&\qquad \qquad \forall \I\subset \T_k,\forall k:1\leq k\leq K+1.
\end{align}
\end{theorem}

It should be noted that some of the inequalities in (\ref{eq:rate})
may be redundant, allowing for further simplification of
(\ref{eq:rate}). In particular, an inequality in which there exists an
$R_{lj}$ on the right-hand side that does not appear on the
left-hand side of another inequality may be ignored, since such an
$R_{lj}$ is unbounded.

Fig.~\ref{fig:rate_describe} describes an example of rate
constraints produced by this theorem for the  network in
Fig.~\ref{fig:4relays_net}. The rate constraint on the source rate in
Fig.~\ref{fig:rate_describe}  is derived by combining the two
inequalities derived from \eqref{eq:rate} by setting $\J_{\I}=\{\}$
for $k=4$, and $\J_{\I}=\D_4$ for $k=5$.

\section{Examples of the Parity Forwarding Protocols \label{sec:examples}}
 In this section, a number of examples of parity forwarding
protocols are presented. These examples are chosen to illustrate
different aspects of the parity forwarding scheme in a multiple-relay
network, and the conditions under which parity forwarding is
capacity achieving. It is shown by example that previous DF rates
are achieved by appropriately  choosing  the message tree, the
message sets, and the decoding sets.

In the first example, it is demonstrated that the multihop DF rate is
achievable using parity forwarding. This example introduces a
simple form of a message tree:  the chain message tree. The next
two examples are also designed based on the chain message tree to
demonstrate that the parity forwarding protocol improves the
previous multihop DF rate. The difference between these three
protocols also highlights the impact of selecting different decoding
sets on the achieved rate. A new set of degradedness conditions are
found under which the parity forwarding protocols in these examples
are capacity achieving.

In addition, this section continues the two-relay example in Section
\ref{sec:2relays}. For the two-relay network, a parity forwarding
protocol that uses a message tree which is different from the chain
message tree is described. While the chain message tree examples
focus on the impact of the decoding sets on the achievable rate,
this two-relay example illustrates the impact of the message tree
on the achievable rate. Finally, this section concludes with a parity
forwarding example for the single-relay channel to illustrate source
message splitting.

\subsection{Achieving the Multihop Rate}
There are a number of parity forwarding protocols that achieve the
multihop rate of \cite{xie_kumar_rate}\footnote{One such parity
forwarding protocol is described in this section. Another parity
forwarding scheme with the same rate can be devised, for example,
by using a message tree in which the relay messages are
independent direct parities of the source message.}. A simple
message tree that can be used to achieve the multihop rate is the
{\em chain message tree} depicted in Fig.~\ref{fig:chain_tree}, in
which each message is a parity for its parent message. Depending
on which messages in the tree the relays decode or send, i.e.,
depending on the message sets and decoding sets, different parity
forwarding protocols with different rates are obtained. The multihop
rate can be achieved by having the source message $m_0$ decoded
at  all the relays.

\iftwocol
\begin{figure}
\centering
\includegraphics[scale=0.5]{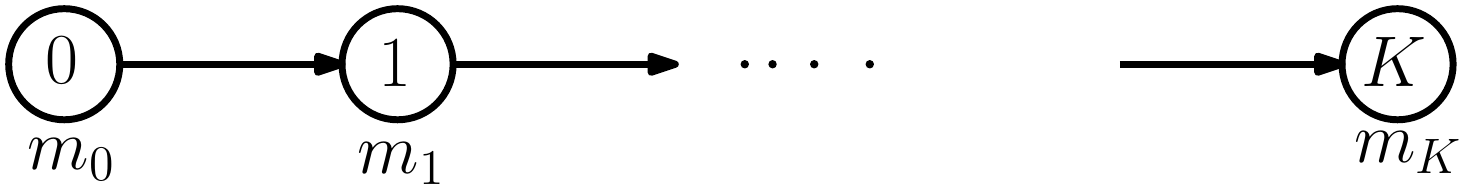}
\caption{Chain message tree.}
\label{fig:chain_tree} 
\end{figure}
\else
\begin{figure}
\centering
\includegraphics[scale=0.7]{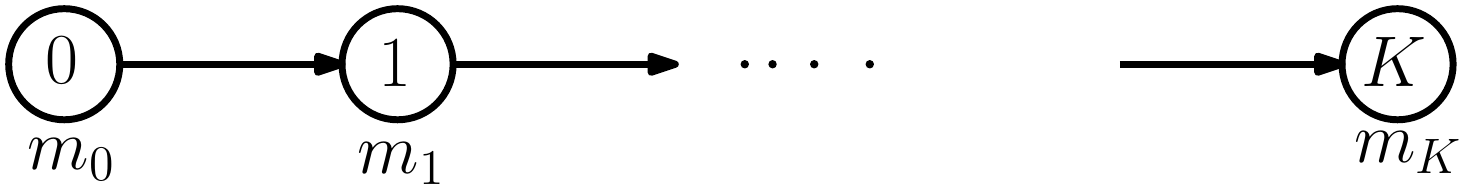}
\caption{Chain message tree.}
\label{fig:chain_tree} 
\end{figure}
\fi

 Consider a network with a source and a destination and $K$
relays indexed by $k=1,\ldots,K$. In multihop parity
forwarding, every relay $k$ decodes the source message $m_{0}$
in the chain message tree in Fig. \ref{fig:chain_tree}, i.e.,
$\D_k=\{m_{0}\}$, and sends a parity message $m_k$ for the
message of the $(k-1)$th relay $m_{k-1}$, i.e., $\A_k=\{m_k\}$.
The known message set at the $k$th relay is given by
$\C^m_k=\{m_{k+1},m_{k+2},\cdots,m_{K}\}$. Hence, the set
$\T_k=\{m_0,m_1,\ldots,m_{k-1}\}$ should be jointly decoded at
the $k$th relay. Associating a random variable $X_k$ with the
message $m_k$ in the message tree, \eqref{eq:rate} in Theorem
\ref{thm:total_rate} can be rewritten for this parity
forwarding protocol as follows:
\begin{subequations}\label{eq:rate:multihop1}
\begin{align}
R_{0}&\leq \sum_{i=0}^{l-1}I(X_i,;Y_k|X_{i+1},\cdots,X_K) +R_l, \quad 0\leq l\leq k-1 \label{eq:rate:multihop1:a}\\
&{=}I(X_0,X_1,\ldots,X_{l-1};Y_k|X_{l},X_{l+1},\ldots,X_K)+R_l \label{eq:rate:multihop1:b}\\
R_{0}&\leq \sum_{i=0}^{k-1}I(X_i,;Y_k|X_{i+1},\cdots,X_K)  \label{eq:rate:multihop1:c}\\
&{=}I(X_0,X_1,\ldots,X_{k-1};Y_k|X_k,X_{k+1},\ldots,X_K)&\label{eq:rate:multihop1:d}
\end{align}
\end{subequations}
for each $k$, $1\leq k\leq K+1$. Inequalities in
\eqref{eq:rate:multihop1:a} are derived from \eqref{eq:rate} for
subsets $\I$ of $\T_k$ for which
$\J_{\I}=\{m_l\},\I'=\{m_0,m_1,\ldots,m_{l-1}\} $. Similarly,
\eqref{eq:rate:multihop1:c} is derived from \eqref{eq:rate} for
$\I=\{\},\I'=\{m_0,m_1,\ldots,m_{k-1}\}$.

Now, note that in \eqref{eq:rate:multihop1}, all rates $R_k$ for
$1\leq k\leq K$ appear only on the right-hand side and thus are
unbounded. Hence, \eqref{eq:rate:multihop1} can be simplified by
ignoring those inequalities with an $R_k$ on their right-hand side to
achieve the following rate which is equal to the multihop rate of
\cite{xie_kumar_rate}:

\begin{subequations}\label{eq:rate:multihop}
\begin{align}
R_0&\leq I(X_0,X_1,\ldots,X_{k-1};Y_k|X_{k},\ldots,X_K)& 1\leq k\leq K+1 \nonumber
\end{align}
\end{subequations}

%

\subsection{Short-Range Relays}
A different choice of decoded messages at the relays results in a
different achievable rate in the previous example. For the chain
message tree shown in Fig.~\ref{fig:chain_tree}, each relay may
only decode the message of its predecessor, i.e., setting
$\D_k=\{m_{k-1}\}$, and send a parity for $m_{k-1}$, i.e.,
$\A_k=\{m_k\}$ same as in the previous example. This would be a
good  scheme if each relay terminal has a small range  and  can
only communicate to its successor relay (e.g., the degraded network
shown in Fig. \ref{fig:chain_network1}).

Theorem \ref{thm:total_rate} can be used to give the achievable
rate of this parity forwarding protocol. By setting
$\D_k=\{m_{k-1}\}$, $1\leq k\leq K$, the $k$th relay in each block
can compute parity messages with orders greater than $k$, i.e.,
$\C^m_k=\{m_{k+1},m_{k+2},\cdots,m_{K}\}$, $0\leq k\leq K$.
The only message with  an order less than $k$ that is parity of
$\D_k$ is $m_{k-1}$, hence, $\T_k=\{m_{k-1}\}$, $0<k\leq K$,
$1<k\leq K$. For the destination, $\D_{K+1}=\{m_0\}$, and
$\T_{K+1}=\{m_0,m_1,\ldots,m_{K}\}$.

Now, for $1\leq k\leq K$, the only subsets $\I$ of
$\T_k=\{m_{k-1}\}$ are $\{\}$ and $\T_k$. Hence, the inequalities
corresponding to $1\leq k\leq K$ in (\ref{eq:rate}) are given by
\begin{subequations}\label{eq:rate:short_range:explain}
\begin{align}
R_{k-1}&\leq I(X_{k-1};Y_k|X_k,\cdots,X_K)& & :\I=\{\},\I'=\{m_{k-1}\} \\
R_{k-1}&\leq R_{k-1} & &:\I=\{m_{k-1}\},\I'=\{\}.
\end{align}
\end{subequations}
Finally, since $\T_{K+1}=\{m_0,m_1,\cdots,m_{K}\}$ and
$\D_{K+1}=\{m_0\}$, \eqref{eq:rate} gives the following
inequalities for $k=K+1$:
\begin{subequations}\label{eq:rate_K+1}
\begin{align}
R_0&\leq \sum_{i=0}^{l-1} I(X_i;Y_{K+1}|X_{i+1},\cdots,X_{K}) + R_l, 1\leq l \leq K\\
R_0&\leq \sum_{i=0}^{K} I(X_i;Y_{K+1}|X_{i+1},\cdots,X_{K}).
\end{align}
\end{subequations}
The above inequalities are derived by setting $\J_{\I}=\{m_l\}$
and $\I'=\{m_0,\cdots,m_{l-1}\}$ for $1\leq l\leq K$, and $\J_{\I}=\{\}$ for $l=K+1$.

\begin{figure}
  \begin{center}
  \iftwocol
        {\scalebox{1}{\includegraphics[scale=0.4]{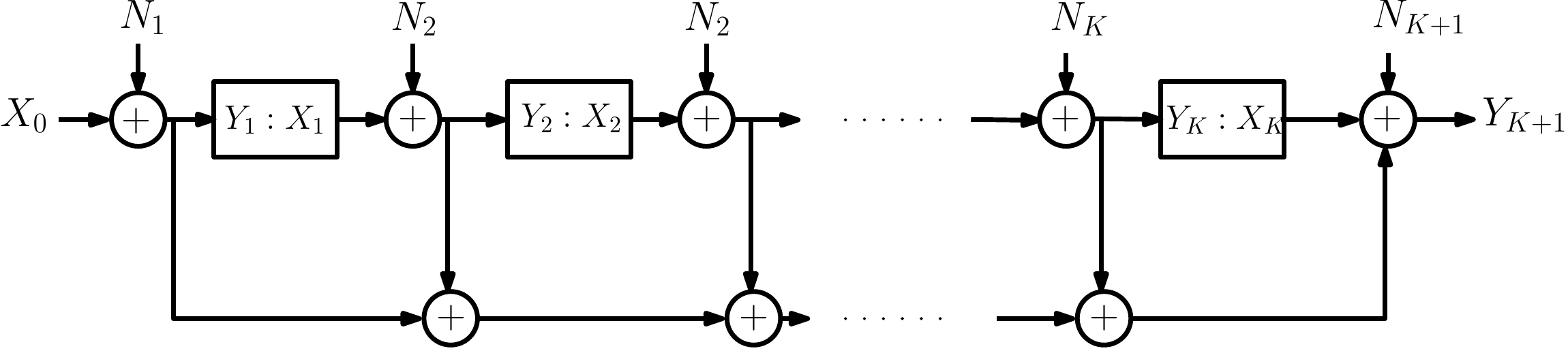}}}
        \else
        {\scalebox{1}{\includegraphics[scale=0.7]{g_network1.pdf}}}
        \fi
        \caption{\label{fig:chain_network1}A degraded chain network with additive noises $N_k$, $1\leq k\leq K$. Each relay can directly communicate only to its successor terminal.
         The $k$th relay decodes the message of relay $k-1$ and forwards a parity for it.}
  \end{center}
\end{figure}

Using the chain rule for mutual information and substituting the
constraints on $R_l$ for $1\leq l\leq K$ in \eqref{eq:rate_K+1}
results in an achievable rate for this network specified in the next
theorem.
\begin{theorem}\label{thm:short_range}
For a memoryless $K$-relay network defined by
$p(y_1,y_2,\ldots,y_{K+1}|x_0,x_1,\ldots,x_K)$, fixing any
$p(x_0,x_1,\ldots,x_K)$, the source rate $R_0$ satisfying the
following constraints is achievable:
\begin{align}\label{eq:chain:1:rate}
R_0&<I(X_0;Y_1|X_1,X_2,\ldots,X_K)\\
R_0&<I(X_0;Y_{K+1}|X_1,\ldots,X_K)+I(X_1;Y_2|X_2,\ldots,X_K)\nonumber\\
R_0&<I(X_0,X_1;Y_{K+1}|X_2,\ldots,X_K)+I(X_2;Y_3|X_3,\ldots,X_K)\nonumber\\
&\vdots\nonumber\\
R_0&<I(X_0^{K-2};Y_{K+1}|X_{K-1},X_K)+I(X_{K-1};Y_K|X_K)\nonumber\\
R_0&<I(X_0^K;Y_{K+1})\nonumber
\end{align}
where $X_i^j\triangleq (X_i,X_{i+1},\ldots,X_j)$. Further, the
above rate maximized over $p(x_0,x_1,\ldots, x_K)$ is the capacity
of this network if $X_k-(Y_{k+1},X_{k+1}^K)-Y_{k+2}^{K+1}$
and $X_0^{k-1}-(Y_{K+1},X_k^K)-Y_{k+1}^{K}$ form Markov
chains.
\end{theorem}
\begin{proof}
The achievability follows from the statements leading to the
theorem. The converse can be proved using the cut-set bound
\cite[Theorem 14.10.1]{cover_elements}.

The cut-set bound states that the source rate $R_0$ is upper
bounded by the following inequalities for $0\leq k\leq K$:
\begin{align}\label{eq:short_range:cut-set}
R_0&<I(X_0,\ldots,X_k;Y_{k+1},\ldots,Y_{K+1}|X_{k+1},\ldots,X_{K})
\end{align}
for some $p(x_0,x_1,\ldots, x_K)$. The above upper bound coincides
with the achievable rate in \eqref{eq:chain:1:rate} if
$X_k-(Y_{k+1},X_{k+1}^K)-Y_{k+2}^{K+1}$ and
$X_0^{k-1}-(Y_{K+1},X_k^K)-Y_{k+1}^{K}$ form Markov chains.
This can be proved by expanding \eqref{eq:short_range:cut-set} as
follows:
\begin{subequations}
\begin{align}
&I(X_0^k;Y_{k+1}^{K+1}|X_{k+1}^K)\nonumber\\
&\overset{(a)}{=}I(X_k;Y_{k+1}^{K+1}|X_{k+1}^K)+I(X_0^{k-1};Y_{k+1}^{K+1}|X_{k}^K)\nonumber\\
&\overset{(b)}{=}I(X_k;Y_{k+1}|X_{k+1}^K)+I(X_0^{k-1};Y_{k+1}^{K+1}|X_k^K)\nonumber\\
&\overset{(c)}{=}I(X_k;Y_{k+1}|X_{k+1}^K)+I(X_0^{k-1};Y_{K+1}|X_k^K)\nonumber,
\end{align}
\end{subequations}
where (a) follows from the chain rule for mutual information, (b)
holds since $X_k-(Y_{k+1},X_{k+1}^K)-Y_{k+2}^{K+1}$ is a
Markov chain, and (c) holds since
$X_0^{k-1}-(Y_{K+1},X_k^K)-Y_{k+1}^{K}$ is a Markov chain.
Fig.~\ref{fig:chain_network1} shows an example of such a degraded
network with additive noise at receivers.
\end{proof}

Intuitively, each rate constraint in the above achievable rate
consists of two components: the $I(X_{k};Y_{k+1}|X_{k+1}^K)$
component represents the transferable information from the $k$th
relay terminal to its successor, and the other component
$I(X_0^{k-1};Y_{K+1}|X_k^K)$ corresponds to the rate at which
source and the first $k-1$ relays  can cooperatively communicate to
the destination. The degradedness conditions in Theorem
\ref{thm:short_range} ensure that  the achievable rate coincides
with the cut-set defined at the $k$th relay separating the source
and the first $k$ relays from the destination and the rest of the
relays.

The rate achieved by this protocol can be higher than the multihop
rate if relays have a short range and  the channel from the source
to  relays that are far away  is blocked. This example generalizes
Protocol B introduced in Section \ref{sec:2relays}. It demonstrates
that the best choice of decoding sets at the relays depends on the
network condition. In general, the best achievable parity forwarding
rate is obtained by searching through all possible parity forwarding
protocols for the network.

\subsection{Coupled Relays}
It is possible to group the relays to cooperate in the previous
example.  The example in this section describes a parity forwarding
protocol based on the same chain message tree introduced in the
previous example, but in which the relays are coupled in groups of
two, cooperatively  communicating to the next relay. In the network
shown in Fig.~\ref{fig:chain_network2}, the source and the first
relay cooperatively communicate to the second relay. The second
relay and the third relay cooperatively communicate to the fourth
relay, and so on. Assume that the number of relays $K$ is odd.
Similar to the previous examples,  the $k$th relay sends a parity
message for the message of relay $k-1$, i.e., $\A_k=\{m_k\}$. The
difference as compared to the previous examples, is the choice of
messages that are decoded at relays. For odd $k$'s, the $k$th relay
decodes the message of relay $k-1$, and for even $k$'s, it decodes
the message of relay $k-2$. This scenario corresponds to setting
$\D_k=\{m_{k-1}\}$ for odd $k$, and $\D_k=\{m_{k-2}\}$ for even
$k$. Hence, $\T_k=\{m_{k-1}\}$ for odd $k$, and
$\T_k=\{m_{k-2},m_{k-1}\}$ for even $k$ (see
Fig.~\ref{fig:chain_rate}). For the destination,
$\T_{K+1}=\{m_0,\cdots,m_{K}\}$, and $\D_{K+1}=\{m_0\}$.
The following theorem specifies  an achievable rate for this setting.

\begin{figure}
\iftwocol
{\scalebox{1}{\includegraphics[scale=0.39]{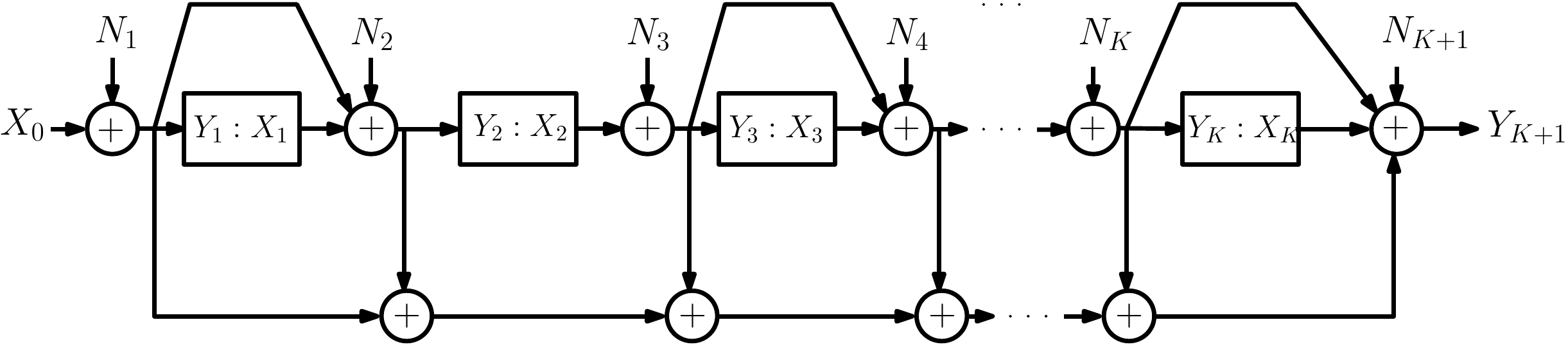}}}
\else
{\scalebox{1}{\includegraphics[scale=0.7]{g_network2.pdf}}}
\fi
\caption{\label{fig:chain_network2} A degraded  chain network with
additive noise. For odd $k$'s, the $k$th relay, $1\leq k\leq K$,
helps relay $k+1$ decode the message of relay $k-1$.}
\end{figure}

\begin{figure}[t]
  \begin{center}
  \iftwocol
        {\scalebox{1}{\includegraphics[scale=0.4]{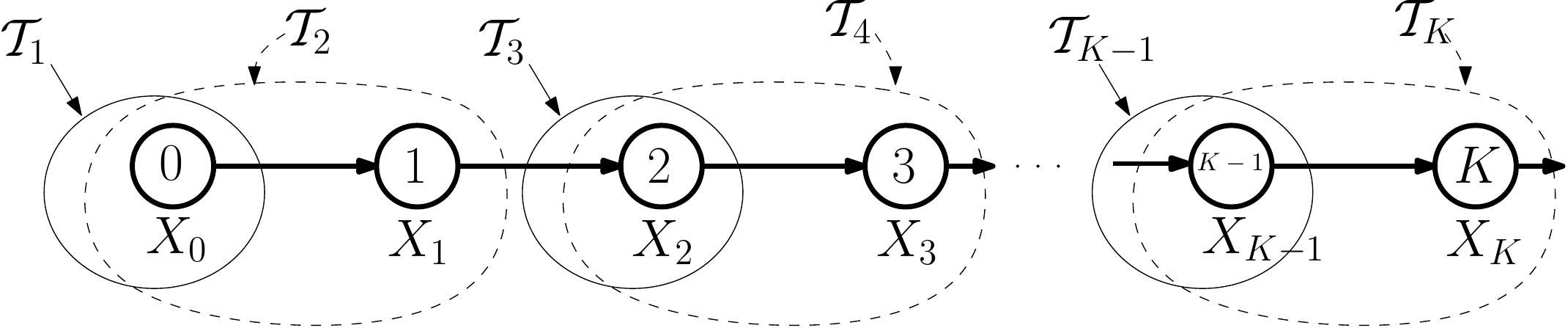}}}
        \else
        {\scalebox{1}{\includegraphics[scale=0.5]{chain_rate.pdf}}}
        \fi
        \caption{\label{fig:chain_rate}The sets $\T_k$ for the example in Fig.~\ref{fig:chain_network2}.}
  \end{center}
\end{figure}

\begin{theorem}\label{thm:coupled_range}
For a memoryless $K$-relay network, $K$ odd, defined by
$p(y_1,y_2,\ldots,y_{K+1}|x_0,x_1,\ldots,x_K)$, fixing any
$p(x_0,x_1,\ldots,x_K)$, the source rate $R_0$ satisfying the
following constraints is achievable: \iftwocol
\begin{align}
R_0&<I(X_0;Y_1|X_1^K)\label{eq:couple_rate}\\
R_0&<I(X_0,X_1;Y_2|X_2^K)\nonumber\\
R_0&<I(X_2;Y_3|X_3^K)+I(X_0,X_1;Y_{K+1}|X_2^K)\nonumber\\
R_0&<I(X_2,X_3;Y_4|X_4^K)+I(X_0,X_1;Y_{K+1}|X_2^K).\nonumber\\
&\vdots\nonumber\\
R_0&<I(X_{K-2};Y_{K-1}|X_{K-1},X_K)+I(X_0^{K-3};Y_{K+1}|X_{K-2}^K)\nonumber\\
R_0&<I(X_{K-2}^{K-1};Y_K|X_K)+ I(X_0^{K-3};Y_{K+1}|X_{K-2}^K) \nonumber \\
R_0&<I(X_0,X_1,\cdots,X_K;Y_{K+1}).\nonumber
\end{align}
\else{
\begin{align}
R_0&<I(X_0;Y_1|X_1,\cdots,X_K)\label{eq:couple_rate}\\
R_0&<I(X_0,X_1;Y_2|X_2,\cdots,X_K)\nonumber\\
R_0&<I(X_2;Y_3|X_3,\cdots,X_K)+I(X_0,X_1;Y_{K+1}|X_2,\cdots,X_K)\nonumber\\
R_0&<I(X_2,X_3;Y_4|X_4,\cdots,X_K)+I(X_0,X_1;Y_{K+1}|X_2,\cdots,X_K).\nonumber
\end{align}\vspace{-0.7cm}
\begin{align}
&\vdots\nonumber\\
R_0&<I(X_{K-2};Y_{K-1}|X_{K-1},X_K)+\nonumber\\
& \qquad \quad \qquad  I(X_0,X_1,\cdots,X_{K-3};Y_{K+1}|X_{K-2},\cdots,X_K)\nonumber\\
R_0&<I(X_{K-2},X_{K-1};Y_K|X_K)+\nonumber\\
&\qquad \qquad \qquad I(X_0,X_1,\cdots,X_{K-3};Y_{K+1}|X_{K-2},X_{K-1},X_K) \nonumber \\
R_0&<I(X_0,X_1,\cdots,X_K;Y_{K+1}).\nonumber.
\end{align}
\fi Further, the above rate maximized over $p(x_0,x_1,\ldots,x_K)$
is the capacity of this network if
$X_k-(Y_{k+1},X_{k+1}^K)-Y_{k+2}^K$ and
$X_0^{k-1}-(Y_{K+1},X_k^K)-Y_{k+1}^{K}$ are Markov chains for
$k$ even, and
$(X_{k-1},X_k)-(Y_{k+1},X_{k+1}^K)-Y_{k+2}^{K+1}$ and
$X_0^{k-2}-(Y_{K+1},X_{k-1}^K)-Y_{k+1}^{K}$ form Markov
chains for $k$ odd.
\end{theorem}

\begin{proof}
This rate is derived using \eqref{eq:rate} as follows. For $k$ odd,
$\T_k=\{m_{k-1}\}$ has only two subsets $\{\}$ and $\T_k$. Thus,
\eqref{eq:rate} for odd $k$, $1\leq k\leq K$, results in:
\begin{align}
R_{k-1}&\leq I(X_{k-1};Y_k|X_k,\cdots,X_K)& & :\I=\{\},\I'=\{m_{k-1}\} \nonumber\\
R_{k-1}&\leq R_{k-1} & &:\I=\{m_{k-1}\},\I'=\{\}.\nonumber
\end{align}
For even $k$, $\T_k=\{m_{k-2},m_{k-1}\}$ which has subsets
$\{\}$, $\{m_{k-2}\}$, $\{m_{k-1}\}$, and $\T_k$. Consequently,
the corresponding constraints in \eqref{eq:rate} for even $k$,
$2\leq k<K$, state that
\begin{subequations}
\iftwocol
\begin{align}
\intertext{for $\I=\{\},\I'=\T_k:$} 
\quad R_{k-2}&\leq I(X_{k-2};Y_k|X_{k-1},\cdots,X_K)\nonumber\\
&\qquad\qquad\qquad+I(X_{k-1};Y_k|X_k,\ldots,X_K),\\ \intertext{for $\I=\{m_{k-2}\},\I'=\{m_{k-2}\}$:}
\quad R_{k-2}&\leq I(X_{k-2};Y_k|X_{k-1},\cdots,X_K)+R_{k-1},\\
\intertext{and for $\I=\T_k,\I'=\{\}:$}
\quad R_{k-2}&\leq R_{k-2}
\end{align}
\else
\begin{align}
R_{k-2}&\leq I(X_{k-2};Y_k|X_k,\cdots,X_K)+I(X_{k-1};Y_k|X_k,\ldots,X_K)& & :\I=\{\},\I'=\T_k \\
R_{k-2}&\leq I(X_{k-2};Y_k|X_k,\cdots,X_K)+R_{k-1} & &:\I=\{m_{k-2}\},\I'=\{m_{k-2}\}\\
R_{k-2}&\leq R_{k-2} & &:\I=\T_k,\I'=\{\}.
\end{align}
\fi
\end{subequations}
Finally, the rate constraints at the destination are also given by a
set of inequalities given in \eqref{eq:rate_K+1}.  The rate given in
\eqref{eq:couple_rate} is obtained by using the chain rule for
mutual information and ignoring constraints that have a rate $R_k$
for an odd $k$ on their right-hand sides, as only  $R_k$'s with an
even $k$ are constrained by the above set of bounds.

The converse is proved using the cut-set bound. For even $k$'s,
$1\leq k\leq K$, the cut-set bound results in a set of upper bounds
similar to those in \eqref{eq:short_range:cut-set} which are
achievable using this protocol if  for $k$ even, $0\leq k\leq K-1$,
$X_k-(Y_{k+1},X_{k+1}^K)-Y_{k+2}^K$  and
$X_0^{k-1}-(Y_{K+1},X_k^K)-Y_{k+1}^{K}$ are Markov chains.
For odd $k$'s, the cut-set bound coincides with
\eqref{eq:couple_rate} if
$(X_{k-1},X_k)-(Y_{k+1},X_{k+1}^K)-Y_{k+2}^{K+1}$ and
$X_0^{k-2}-(Y_{K+1},X_{k-1}^K)-Y_{k+1}^{K}$ form Markov
chains. This is proved in the following:
\begin{subequations}\label{eq:coupled:cut-set}
\begin{align}
R_0&<I(X_0,\ldots,X_k;Y_{k+1},\ldots,Y_{K+1}|X_{k+1},\ldots,X_{K})\\
&\overset{(a)}{=}I(X_{k-1},X_k;Y_{k+1}^{K+1}|X_{k+1}^K)+I(X_0^{k-2};Y_{k+1}^{K+1}|X_{k-1}^K)\\
&\overset{(b)}{=}I(X_{k-1},X_k;Y_{k+1}|X_{k+1}^{K+1})+I(X_0^{k-2};Y_{k+1}^{K+1}|X_{k-1}^K)\\
&\overset{(c)}{=}I(X_{k-1},X_k;Y_{k+1}|X_{k+1}^{K+1})+I(X_0^{k-2};Y_{K+1}|X_{k-1}^K),
\end{align}
\end{subequations}
where (a) follows from the chain rule for mutual information, (b)
holds since
$(X_{k-1},X_k)-(Y_{k+1},X_{k+1}^K)-Y_{k+2}^{K+1}$ forms a
Markov chain for odd $k$, $1\leq k\leq K$, and (c) holds since
$X_0^{k-2}-(Y_{K+1},X_{k-1}^K)-Y_{k+1}^{K}$ is a Markov
chain.
\end{proof}

An example of such a degraded network is shown in
Fig.~\ref{fig:chain_network2}. For this network,
\eqref{eq:couple_rate} achieves the cut-set bound  if we consider
the cut-set defined at the $k$th relay, separating the source and
the first $k$ relays from the destination and relays $k+1$ up to
$K$. For an odd $k$, the rate of this cut-set  equals to the
cooperative information rate from relays $k-1$ and $k$  to relay
$k+1$, plus the cooperative rate from the source and relays 1 up to
$k-2$ to the destination. For an even $k$, the rate of this cut-set is
given by the rate at which the $k$th terminal can communicate to
relay $k+1$, plus the rate at which the source and the first $k-1$
relays can communicate to the destination.

\subsection{Two-Relay Network with a Semideterministic Subnetwork}
The previous examples explain the effect of  decoding different
messages  at relays on the achieved rate. The form of the message
tree also affects the achievable rate of parity forwarding. In this
example another parity forwarding protocol is introduced for the
two-relay network of Section \ref{sec:2relays} with a different type
of message tree.

Consider the two-relay network defined in Section \ref{sec:2relays}.
 A parity forwarding protocol can
be devised for this network based on the message tree shown in
Fig.~\ref{fig:semi_det}. In this parity forwarding protocol, the first
relay decodes the source message and transmits two independent
parity messages for the source message. The second relay  decodes
one of the two parity messages sent by the first relay.

More precisely,  we have $\A_0=\{m_0\},
\A_1=\{m_{11},m_{12}\},\A_2=\{m_{21}\}$. The first relay
decodes $m_0$, i.e.,  $\D_1=\{m_0\}$, and the second relay
decodes $m_{11}$, i.e., $\D_2=\{m_{11}\}$. Consequently, the
source knows all other messages, and the first relay knows
$m_{21}$. Hence, $\C^x_{0}=\{m_{11},m_{12},m_{21}\}$,
$\C^x_{11}=\{m_{21}\}$,  $\C^x_{12}=\{m_{11},m_{21}\}$
(note that $m_{11}$ is in the known set of $m_{12}$, since
superposition encoding is used for encoding $m_{12}$ on top of
$m_{11}$ and $m_{21}$), and $\C^x_{21}=\{\}$. Since
$\D_1=\{m_0\}$ and $\A_1=\{m_{11},m_{12}\}$, we have
$\T_1=\{m_0\}$. Similarly, since $\D_2=\{m_{11}\}$ and
$\A_2=\{m_{21}\}$, we have $\T_2=\{m_{11}\}$. At the
destination $\D_3=\{m_0\}$, and
$\T_3=\{m_0,m_{11},m_{12},m_{21}\}$. Associating random
variables $X_0,X_{11},X_{12},X_{21}$ with messages
$m_0,m_{11},m_{12},m_{21}$, respectively,
Theorem~\ref{thm:total_rate} gives the following achievable rate
for an arbitrary two-relay network under any fixed distribution
$p(x_0,x_{11},x_{12},x_{21})$  (the detailed derivation is
omitted):
\begin{subequations}\label{eq:sem_det:rate}
\begin{align}
R_0&<I(X_0;Y_1|X_{11},X_{12},X_{21})\label{eq:2r_rateC_b}\\
R_0&<I(X_0,X_{12};Y_3|X_{11},X_{21})+I(X_{11};Y_2|X_{21})\label{eq:2r_rateC_c}\\
R_0&<I(X_0,X_{11},X_{12},X_{21};Y_3)\label{eq:2r_rateC_a}
\end{align}
\end{subequations}
The above rate can be shown to be the capacity of a two-relay
network if the channel from the source to the first relay is stronger
than the channel from the source to the second relay and the
destination, and the channel from the first relay to the second relay
is semideterministic \cite{elgamal_aref}.

Intuitively, this protocol is suitable for this network if we recall that
{\em partial} decoding at the relay is the optimal strategy for the
semideterministic single-relay channel.  In this channel, the
capacity-achieving strategy is for the source to split its message
into two messages \cite{elgamal_aref}; the relay decodes one of
them and forwards a parity message for it. Similar to the
single-relay case, in a multirelay network with a semideterministic
channel from the first relay to the second relay, the first relay
should split its message into two parts; the second relay only
partially decodes the message of the first relay. This is optimal as
shown in the next theorem.

\begin{figure}[t]
  \begin{center}
  \iftwocol
      \scalebox{0.5}{\includegraphics{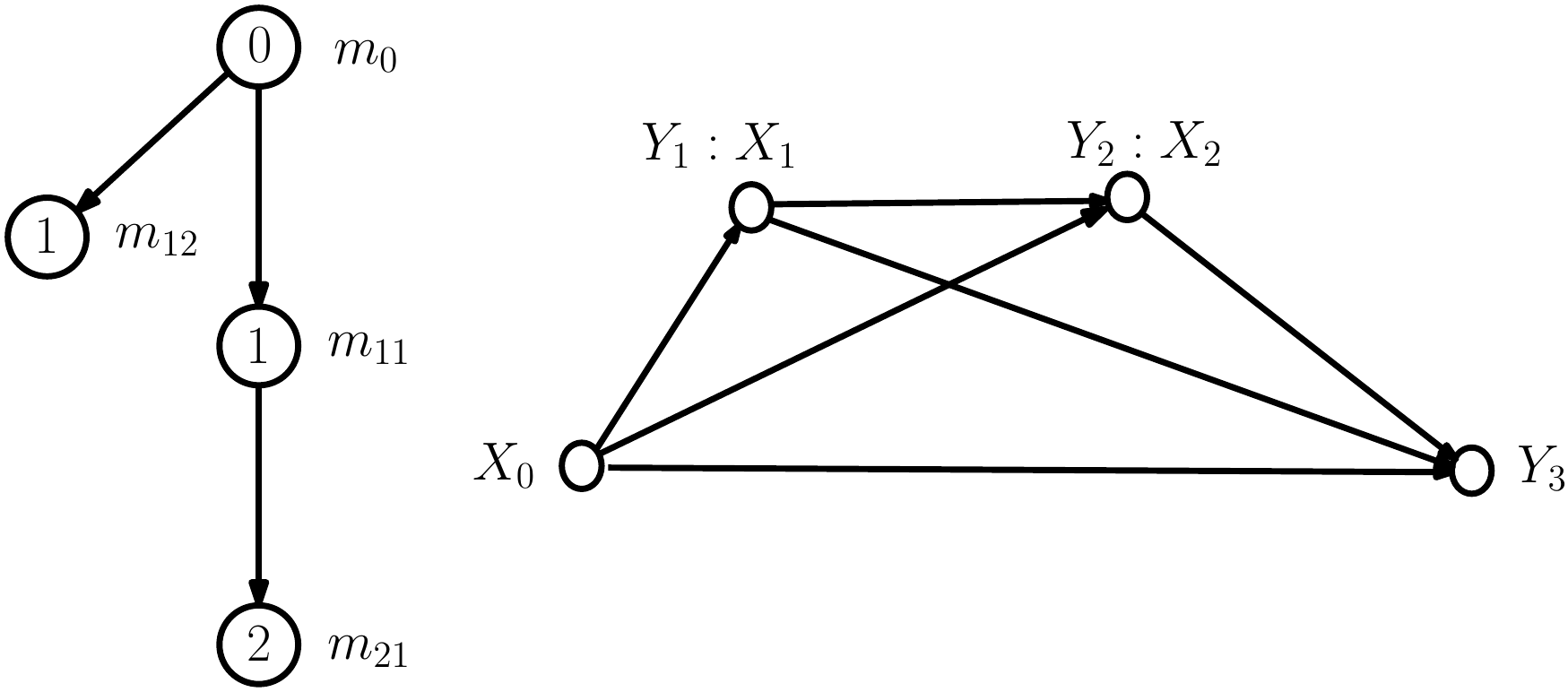}}
      \else
      \scalebox{0.7}{\includegraphics{2relay_tree3.pdf}}
      \fi
        \caption{Message tree for two-relay network with semideterministic
        channel from the first relay to the second
relay.\label{fig:semi_det}}
  \end{center}
\end{figure}

\begin{theorem}\label{thm:2r_deterministic_rate}
The  capacity of a two-relay network defined by
$p(y_1,y_2,y_3|x_0,x_1,x_2)$ in which $X_0-(X_1,X_2,Y_1)-Y_3$ is
a Markov chain, and the channel from the first relay to the second
relay is semideterministic, i.e., $Y_2=f(X_1,X_2)$ for a
deterministic function $f(\cdot , \cdot)$, is given by:
\begin{subequations}\label{eq:2r_det_rateC}
\begin{align}
R_0&<I(X_0;Y_1|X_1,X_2)\label{eq:2r_det_rateC_a}\\
R_0&<I(X_0;Y_3|X_1,X_2)+I(X_1;Y_3|Y_2,X_2)+H(Y_2|X_2)\label{eq:2r_det_rateC_b}\\
R_0&<I(X_0,X_1,X_2;Y_3)\label{eq:2r_det_rateC_c}
\end{align}
\end{subequations}
maximized over $p(x_0,x_1,x_2)$.
\end{theorem}
\begin{proof}
For achievability, setting  $X_0=X_0$, $X_1=X_{12}$,
$X_2=X_{21}$,  and $X_{11}=Y_2$  reduces
\eqref{eq:sem_det:rate} to \eqref{eq:2r_det_rateC}. Note that
$X_{11}$ is a deterministic function of $X_{12}$ and $X_{21}$.
Using this fact, the derivation of \eqref{eq:2r_det_rateC_a} and
\eqref{eq:2r_det_rateC_c} is straightforward. The achievability of
\eqref{eq:2r_det_rateC_b} is proved in the following:
 \iftwocol
\begin{subequations}
\begin{align*}
&I(X_0,X_{12};Y_3|X_{11},X_{21})+I(X_{11};Y_2|X_{21})\nonumber\\
&\overset{(a)}{=}I(X_0;Y_3|X_{11},X_{12},X_{21})+ I(X_{12};Y_3|X_{11},X_{21})+\nonumber\\
& \qquad I(X_{11};Y_2|X_{21})\nonumber\\
&\overset{(b)}{=}I(X_0;Y_3|X_{12},X_{21})+I(X_{12};Y_3|X_{11},X_{21})\nonumber\\
&\qquad+ I(Y_2;Y_2|X_{2})\nonumber\\
&\overset{(c)}{=} I(X_0;Y_3|X_{1},X_2)+I(X_{1};Y_3|Y_{2},X_{2})+H(Y_2|X_2)\nonumber,
\end{align*}
\end{subequations}
\else
\begin{subequations}
\begin{align*}
I(X_0,X_{12};Y_3|X_{11},X_{21})&\overset{(a)}{=}I(X_0;Y_3|X_{11},X_{12},X_{21})+I(X_{12};Y_3|X_{11},X_{21})+ I(X_{11};Y_2|X_{21})\\
&\overset{(b)}{=}I(X_0;Y_3|X_{12},X_{21})+I(X_{12};Y_3|X_{11},X_{21})+\nonumber\\
&\qquad \qquad \qquad\qquad\qquad \qquad I(Y_2;Y_2|X_2)\nonumber\\
&\overset{(c)}{=} I(X_0;Y_3|X_{1},X_2)+I(X_{1};Y_3|Y_{2},X_{2})+H(Y_2|X_2),
\end{align*}
\end{subequations}\fi
where (a) follows from the chain rule for mutual information, (b)
follows since $X_{11}$ is a function of $X_{12}$ and $X_{21}$,
and (c) follows from $I(Y_2;Y_2|X_2)=H(Y_2|X_2)$.

To prove the converse, we use the cut-set bound
\eqref{eq:cut-set}. The Markov chain $X_0-(X_1,X_2,Y_1)-Y_3$
along with the condition $Y_2=f(X_1,X_2)$ leads to the Markov
chain $X_0-(X_1,X_2,Y_1)-(Y_2,Y_3)$ for this network. Hence, the
bound \eqref{eq:cut-set:1}  coincides with
\eqref{eq:2r_det_rateC_a}. The bound \eqref{eq:cut-set:3} is also
equivalent to \eqref{eq:2r_det_rateC_c}. It remains to prove that
\eqref{eq:2r_det_rateC_b} is achievable and meets the cut-set
bound (\ref{eq:cut-set:2}). This can be proved by expanding
(\ref{eq:cut-set:2}) as follows: \iftwocol
\begin{align}
R_0&<I(X_0,X_1;Y_2,Y_3|X_{2})\nonumber\\
&=I(X_0;Y_2,Y_3|X_1,X_2)+I(X_1;Y_2,Y_3|X_2)\nonumber\\
&=I(X_0;Y_3|X_1,X_2,Y_2)+I(X_0;Y_2|X_1,X_2)+\nonumber\\
& \qquad  I(X_1;Y_3|X_{2},Y_2)+I(X_1;Y_2|X_2)\nonumber\\
&=
I(X_0;Y_3|X_1,X_2,Y_2)+I(X_1;Y_3|X_{2},Y_2)+H(Y_2|X_2),\label{eq:2r_det_boundc}
\end{align}
\else
\begin{align}
R&\leq I(X_0,X_1;Y_2,Y_3|X_{2})\nonumber\\
&=I(X_0;Y_2,Y_3|X_1,X_2)+I(X_1;Y_2,Y_3|X_2)\nonumber\\
&=I(X_0;Y_3|X_1,X_2,Y_2)+\nonumber\\
& I(X_0;Y_2|X_1,X_2)+ I(X_1;Y_3|X_{2},Y_2)+I(X_1;Y_2|X_2)\nonumber\\
&=
I(X_0;Y_3|X_1,X_2,Y_2)+I(X_1;Y_3|X_{2},Y_2)+H(Y_2|X_2),\label{eq:2r_det_boundc}
\end{align}
\fi since $H(Y_2|X_1,X_2)=0$ and $I(X_0;Y_2|X_1,X_2)=0$ for
$Y_2=f(X_1,X_2)$.
\end{proof}

In addition to the two parity forwarding protocols described in
Section \ref{sec:2relays} and the above example, there are other
possible  parity forwarding protocols for a two-relay network as
well. For example, the source message may also be split to allow
partial decoding at the first relay as well. The next example
illustrates splitting the source message in the single-relay network.

\subsection{Generalized Decode-and-Forward \label{sec:ex:aref}}
Consider the single-relay channel as an example which illustrates
source message splitting in the parity forwarding framework. For
example, the source may send two messages $m_{01}$, $m_{0}$.
The message $m_{01}$ is a random bin index for $m_{0}$. This is
equivalent to splitting $m_{0}$ into two parts. The relay may only
decode $m_{01}$ which is of a lower rate. This strategy, which is
similar\footnote{The scheme described here where the source
message split using a parity message is slightly different from the
generalized DF approach of \cite[Theorem~7]{cover_elgamal},
where the two source messages are independent of each other.
However, the two schemes result in the same rate.} to the
generalized DF \cite[Theorem~7]{cover_elgamal}, increases the DF
rate for example in semideterministic relay channel
\cite{elgamal_aref} or a relay channel with an orthogonal
source-relay channel \cite{elgamal_zahedi}.

\begin{figure}
\centering
\includegraphics[scale=0.5]{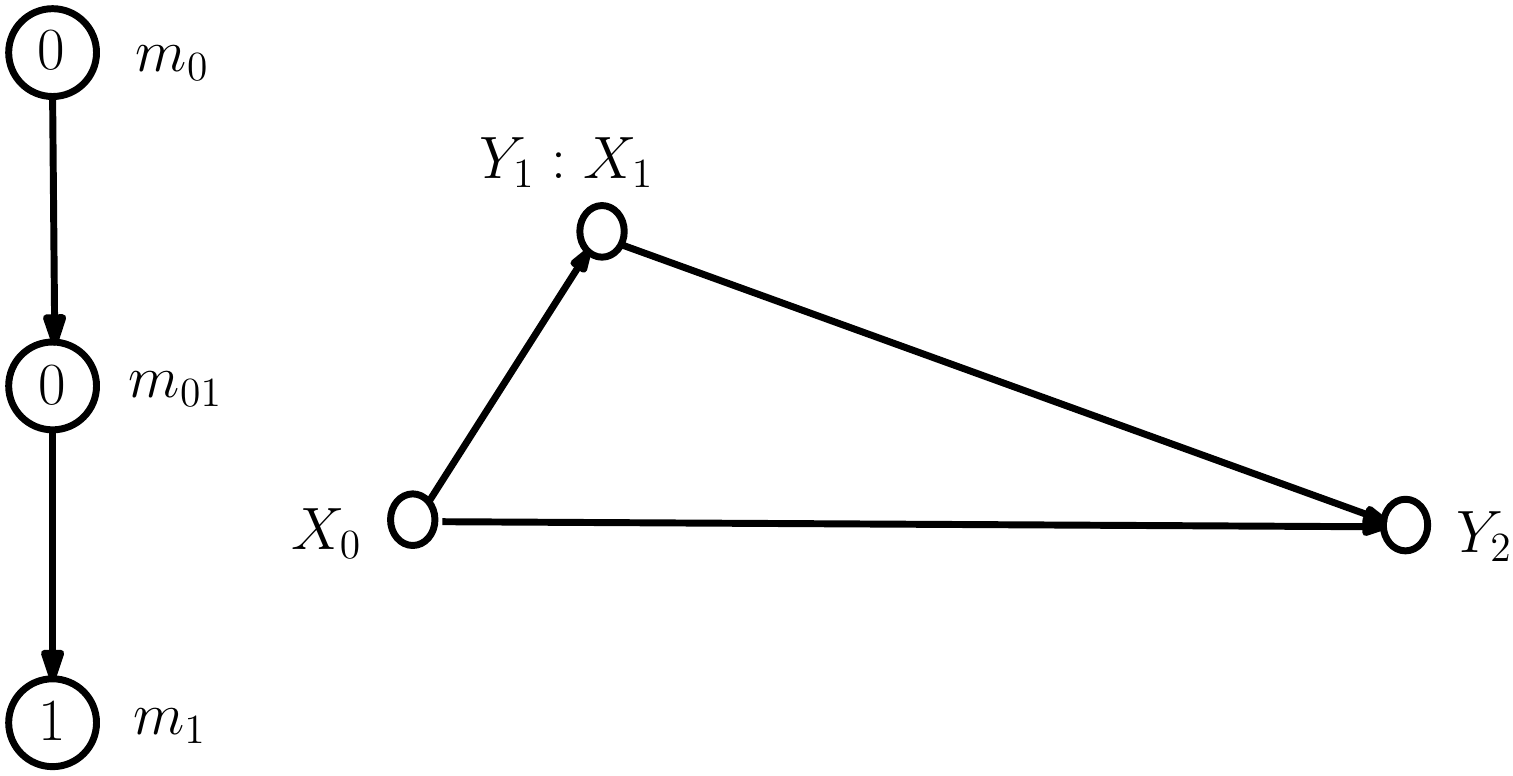}
\caption{The message tree for single-relay generalized decode-and-forward.\label{fig:2relays_tree}}
\end{figure}

The message tree used for this single-relay network is shown in
Fig.~\ref{fig:2relays_tree}, which is similar to the message tree for
the two-relay examples in Section \ref{sec:2relays}. The difference
lies in the way the messages are associated with the source and
the relay for the single-relay network. In this case, the message
sets are defined as follows. The source sends both $m_{0}$ and
$m_{01}$, i.e., $\A_0=\{m_{0},m_{01}\}$. The relay decodes only
$m_{01}$, i.e., $\D_1=\{m_{01}\}$, and sends a parity for
$m_{01}$, i.e., $\A_1=\{m_1\}$. Thus, the known message sets for
each message are given by $\C^m_1=\{\}$,
$\C^m_{01}=\{m_1\}$, and $\C^m_{0}=\{m_1,m_{01}\}$.
Superposition encoding is used to encode messages at the source
and at the relay by associating the random variables $X_1$,
$X_{01}$, and $X_{0}$ with messages $m_1,m_{01}$, and $m_0$,
respectively. The relay encodes $m_1$ by forming a random
codebook of size $2^{nR_1}$. The source first encodes $m_{01}$
superimposed on $m_1$, and then superimposes $m_0$ on top of
$m_{01}$ and $m_1$.

Theorem \ref{thm:total_rate} gives the following rate constraints
for this protocol. At the relay,
\begin{subequations}\label{eq:gdf:rate:1}
\begin{align}
R_{01}&\leq I(X_{01};Y_1|X_1),\label{eq:R_01}\\
\intertext{and at the destination}
R_0&\leq I(X_0;Y_2|X_{01},X_1)+R_{01}\label{eq:R_0:R_01}\\
R_0&\leq I(X_0;Y_2|X_{01},X_1)+I(X_{01};Y_2|X_1)+R_1\label{eq:dgf:2omit}\\
R_0&\leq I(X_0;Y_2|X_{01},X_1)+I(X_{01};Y_2|X_1)+I(X_1;Y_2)\label{eq:chaining}.
\end{align}
\end{subequations}
The above can be simplified by using the chain rule for mutual
information and ignoring \eqref{eq:dgf:2omit} because $R_1$ is
unbounded. Thus, Theorem \ref{thm:total_rate} gives the following
rate maximized over $p(x_0,x_{01},x_1)$  for this protocol:
\begin{subequations}\label{eq:gdf:rate}
\begin{align}
R_0&\overset{(a)}{\leq} I(X_0;Y_2|X_{01},X_1)+I(X_{01};Y_1|X_1)\\
R_0&\overset{(b)}{\leq} I(X_0,X_{01},X_1;Y_2)\nonumber\\
&\overset{(c)}{=}I(X_0,X_1;Y_2).
\end{align}
\end{subequations}
where (a) follows by combining \eqref{eq:R_01} and
\eqref{eq:R_0:R_01}, (b) follows by applying the chain rule for
mutual information to \eqref{eq:chaining}, and (c) follows from
$X_{01}-(X_0,X_1)-Y_2$ as $X_{01}$ can only affect $Y_2$
through the channel inputs $X_0$ and $X_1$. Note that the rate in
\eqref{eq:gdf:rate} is equal to the rate achieved by the generalized
decode-and-forward method (also known as partial
decode-and-forward) of \cite[Theorem 7]{cover_elgamal}.

\section{Conclusions\label{sec:conclusions}}
This paper formulates a class of DF strategies  for an arbitrary
multirelay network. In this set of strategies, called parity
forwarding, relay nodes forward bin indices for the messages of
other transmitters. The message tree structure is utilized to
characterize the encoding and decoding procedures and the
dependencies between messages and their bin indices.  Parity
forwarding improves previous DF strategies because of its
flexibility, and achieves the capacities of new types of degraded
multirelay networks. To derive closed-form expressions for the
achievable rate, we restricted ourselves to the superposition
broadcast encoding. In addition, we also restrict ourselves to a
linear ordering of relays, thus the rates derived in this paper do not
account for the possibility of parallel relaying
\cite{schein_gallager}. Further generalization of this work is
possible by using more advanced broadcast schemes and more
flexible relay topology.


\bibliographystyle{IEEE}
\bibliography{IEEEabrv,reference}

\end{document}